\documentclass[a4paper,10pt]{article}
\usepackage[left=1in, right=1in, top=1in, bottom=1in]{geometry}
\usepackage{wrapfig}
\usepackage{amsmath, amsthm, amssymb}
\usepackage{pdfpages} 
\usepackage{enumitem}
\usepackage[linesnumbered,ruled,vlined]{algorithm2e}
\usepackage{booktabs}
\usepackage{comment}
\usepackage[parfill]{parskip}
\usepackage{graphicx}
\usepackage{epsfig}
\usepackage{bm}
\usepackage{cancel}
\usepackage{wrapfig}
\usepackage{array}
\usepackage{caption}
\usepackage{placeins}
\usepackage{url}
\usepackage{hyperref}
\hypersetup{hidelinks}
\usepackage[bottom]{footmisc}
\usepackage{verbatim}
\usepackage{color,soul,xcolor}

\usepackage{enumitem}
\usepackage{amsmath} %Equation
\usepackage{amsthm} %proofs
\usepackage{amsfonts} % mathbb
\usepackage{mathtools}

\usepackage{fullpage}
\usepackage{parskip} % newlines, no indent
\usepackage{graphicx} % including images
\usepackage{placeins} %FloatBarrier
\graphicspath{ {./images/} }
\usepackage{caption}
\usepackage{subcaption}
\usepackage[T1]{fontenc}

\makeatletter
\renewenvironment{proof}[1][\normalfont\bfseries\proofname]{\par
  \pushQED{\qed}%
  \normalfont \topsep6\p@\@plus6\p@\relax
  \trivlist
  \item\relax
        {\itshape
    #1\@addpunct{.}}\hspace\labelsep\ignorespaces
}{%
  \popQED\endtrivlist\@endpefalse
}

\def\thm@space@setup{%
  \thm@preskip=2\parskip \thm@postskip=0pt
}
\makeatother
\newtheorem{definition}{Definition}[section]
\newtheorem{lemma}{Lemma}[section]
\newtheorem{theorem}{Theorem}[section]
\newtheorem{fact}{Fact}[section]

\newcommand{\K}{\mathcal{K}}
\newcommand{\C}{\mathcal{C}}
\renewcommand{\P}{\mathcal{P}}
\newcommand{\Q}{\mathcal{Q}}
\newcommand{\V}{\mathcal{V}}
\newcommand{\T}{\mathcal{T}}
\newcommand{\I}{\mathcal{I}}
\newcommand{\E}{\mathcal{E}}
\newcommand{\A}{\mathcal{A}}

\newcommand{\R}{\mathcal{R}}
\newcommand{\J}{\mathcal{J}}
\renewcommand{\L}{\mathcal{L}}
\renewcommand{\H}{\mathcal{H}}
\definecolor{myred}{HTML}{B2182B}
\definecolor{myblue}{HTML}{4063FF}

\begin{document}
\title{Sparse dynamic discretization discovery via arc-dependent time discretizations}

\author{Madison Van Dyk\thanks{\noindent Dept. of Combinatorics and Optimization, University of Waterloo, Waterloo, ON, N2L 3G1, Canada. 
\newline \indent ~~Our work was sponsored by the NSERC Discovery Grant Program, grant number RGPIN-03956-2017,
\newline 
\indent~ and by an Amazon post-internship fellowship.} \and Jochen Koenemann$^*$}
\date{}
\maketitle

\begin{abstract}
\parskip=.5\baselineskip

\noindent While many time-dependent network design problems can be formulated as time-indexed formulations with strong relaxations, the size of these formulations depends on the discretization of the time horizon and can become prohibitively large. The recently-developed dynamic discretization discovery (DDD) method allows many time-dependent problems to become more tractable by iteratively solving instances of the problem on smaller networks where each node has its own discrete set of departure times. However, in the current implementation of DDD, all arcs departing a common node share the same set of departure times. This causes DDD to be ineffective for solving problems where all near-optimal solutions require many distinct departure times at the majority of the high-degree nodes in the network. Region-based networks are one such structure that often leads to many high-degree nodes, and their increasing popularity underscores the importance of tailoring solution methods for these networks. 

\noindent To improve methods for solving problems that require many departure times at nodes, we develop a DDD framework where the set of departure times is determined on the arc level rather than the node level. We apply this arc-based DDD method to instances of the service network design problem (SND). We show that an arc-based approach is particularly advantageous when instances arise from region-based networks, and when candidate paths are fixed in the base graph for each commodity. Moreover, our algorithm builds upon the existing DDD framework and achieves these improvements with only benign modifications to the original implementation. 

\end{abstract}

\section{Introduction}

Practical problems arising in supply chain optimization often involve both routing and scheduling decisions. In a typical delivery network, each shipment is transported from its origin to its destination via intermediate warehouses. At each warehouse, shipments may be unloaded from the inbound trailer and loaded onto an outbound trailer. This physical (flat) network can be expressed as a directed graph $D = (N,A)$ with the node set $N$ representing the warehouse locations, and the arc set $A$ representing the available transportation legs between locations. 

In less-than-truckload (LTL) transportation, shipments are small in size relative to the capacity of the trailers, and as a result transportation costs incentivize the consolidation of shipments onto shared trailers. In order to accomplish this consolidation, one must coordinate both the physical route of the shipments as well as the departure time for each leg of the journey. The optimization problem often considered to model this problem is called the \emph{service network design} (SND) problem. In practical applications, there are often additional restrictions on the physical routes of each shipment. For example, in problems considered by Hewitt \cite{enhanced} and Lara et al.~\cite{amazon}, each shipment must be routed on one of a (possibly singleton) set of designated paths (in space). We will refer to SND with the addition that each commodity has a designated feasible subgraph in space as \emph{service network design with restricted routes} (SND-RR). 

Practical temporal problems are solved over some specified time horizon $T$ with an associated time discretization $[T]:= \{0,1, \hdots, T\}$ indicating the times at which decisions can be made (i.e. the times at which trucks can depart warehouses in SND). Temporal network design problems are often modeled using \emph{time-indexed formulations} which have variables and constraints indexed by each time point in a given discretization of the time horizon \cite{SkutellaSurvey}. For example, for an instance of SND where the time horizon is a single day, an hourly discretization would result in a variable for each arc for each hour of the day, representing the possibility for trucks to depart each facility on the hour. These time-indexed formulations can also be thought of as static integer flow formulations in the corresponding \emph{(fully) time-expanded network}, $D_T$, which has copies of each node and arc in $D$ for each time $t \in [T]$. We refer to nodes and arcs in $D_T$ as timed nodes and timed arcs respectively. While time-indexed formulations often have strong LP relaxations, these formulations become impractical to solve for fine time discretizations since the number of variables (and size of $D_T$) grows linearly in $T$ \cite{FordFulkerson58, FordFulkerson59}. 

While an optimal solution may require a large number of departure times in the network, it may be the case that each individual node only requires a small number of departure times. The recently developed \emph{dynamic discretization discovery} (DDD) method of Boland et al.~\cite{Boland1} allows many time-dependent (minimization)\footnote{The DDD framework applies to both minimization and maximization problems. Without loss of generality we will consider minimization problems in this paper.} problems to become more tractable by taking advantage of this observation. A DDD algorithm solves a series of mixed-integer programs (MIPs) defined on \emph{partially time-expanded networks}, denoted $D_S$, that include only a subset of the timed nodes in the fully time-expanded network and shortened timed arcs.  The partially time-expanded networks and corresponding MIPs are constructed to ensure that they provide a lower bound on the optimal value of the original problem. If a solution to the partially time-expanded network cannot be converted to a solution to the original instance of equal cost, the partially time-expanded network is refined by adding timed nodes and timed arcs based on the current solution. In effect, each node begins with a small subset of departure times in $[T]$, and in each iteration the set of departure times at each node is modified non-uniformly until an optimal solution is found. Thus, the DDD framework has the potential to determine which time-points are needed at each node in an optimal solution, without ever constructing the fully time-expanded network. 

In the current implementation of DDD, in each iteration, all arcs departing a common node share the same set of departure times. This is problematic for temporal problems where all near-optimal solutions require a large number of departure times at many nodes in the network. For such problems, in the final iteration of DDD, $D_S$ must contain a large proportion of the timed nodes in $D_T$. Thus, in the current implementation of DDD this leads to a final formulation with a size close to that of the time-indexed formulation on the fully time-expanded network. However, for many problems, while each node may require a large number of departure times in a near-optimal solution, each \emph{arc} may only require a small number of departure times. For these problems, an optimal solution can be expressed on a time-expanded network with a small number of timed arcs relative to the arc set in the fully time-expanded network. 

Temporal network design problems frequently have optimal solutions with departure times that are ``node-dense'' but ``arc-sparse'' when many of the arcs in the network are incident to high-degree nodes. Such high-degree nodes frequently arise in the region-based hub-and-spoke network structure. In a hub-and-spoke network, locations are divided into regions and  each region is represented by a hub. Shipments that have origins and destinations in two different regions must travel between regions via the hubs, resulting in high-degree nodes that have significant throughput. These region-based networks are popular structures in both ground and air transportation networks. Two noteworthy examples are the distribution networks of UPS and FedEx, which have operated with a hub-and-spoke structure for decades \cite{bowen}. An increased focus on these networks is also timely; in the most recent Letter to Shareholders at Amazon, CEO Andy Jassy announced the recent shift of the fulfillment network from a ``national fulfillment network to a  regionalized network model'' that still allows for national shipments when necessary and that optimizes connections between regions \cite{amazon_regional}. The increasing popularity of these networks underscores the importance of tailoring solution methods for region-based networks.  

\subsection{Our Contributions}
To improve methods for solving problems that require many departure times at nodes, we enhance the DDD paradigm by developing an approach where the set of departure times are determined on the arc level rather than the node level. First, we construct a partition of the arcs in the flat network $D$ so that arcs that can be used by a fixed commodity in a feasible solution, which depart the same node,  must be contained in the same part. This partitioning groups arcs into subsets that have the same set of departure times in each iteration of DDD. In the worst case, the finest possible partition results in the trivial partition where each part is the set of outbound arcs at a fixed node, $v$, denoted $\delta^+(v)$. However, in Section \ref{sec:applications} we show that many graph structures allow for finer arc partitions. With this construction we are able to implement a more conservative refinement step in the DDD algorithm. Specifically, a time point can be added to the set of available departure times for arcs in a single part rather than the full set of arcs in $\delta^+(v)$ as is the case of the original DDD approach. This more effective refinement strategy decreases the size of the resulting formulations in each iteration. 

To build a DDD algorithm based on this arc partition, we create an auxiliary flat network and modify the construction of the lower bound and refinement steps. To form the auxiliary graph, $G$, we replicate each node $v$ in $D$, for each part in the partition of $\delta^+(v)$, which allows us to store the set of departure times for the arcs in that part. This replication process leads to multiple copies of each arc $uv$ in the flat network $D$, roughly equal to the number of parts in the partition of $\delta^+(v)$.

We then create a new time-indexed formulation and DDD algorithm where the flat network is the auxiliary graph $G$. Since the auxiliary network has multiple arcs representing the same underlying arc in the flat network $D$, we cannot simply apply the original DDD algorithm to the resulting routing problem with flat network $G$ instead of $D$. Instead, careful attention must be taken to ensure cost savings obtained by consolidating flow onto a common arc are captured at each step of the DDD algorithm. We will refer to our DDD algorithm based on the auxiliary network as an \emph{arc-based} DDD algorithm, and the original approach as a \emph{node-based} DDD algorithm.

Our arc-based DDD approach broadens the application of the DDD paradigm to problems that require a large number of departure times at each node but not each arc. We apply this arc-based approach to the problem of SND-RR, which is a variant of service network design where each commodity has a designated feasible network in space through which it must be routed. We show that this arc-based approach offers speed-up for region-based networks following a hub-and-spoke structure. We also demonstrate that our arc-based DDD algorithm offers improvement over the node-based approach when applied to problems where shipment paths are fixed in the flat (physical) network, as was the case in \cite{amazon}. Moreover, since our algorithm builds upon the existing DDD framework it achieves these improvements with only benign modifications to the original implementation.

{\color{black} In the worst case, the problem structure only permits the trivial arc partition where each part is the set of outbound arcs at a fixed node. However, even in this case we show that our arc-based DDD approach improves over the node-based approach when release times and deadlines are limited to a handful of node-dependent \emph{critical times}, which naturally emerge from shift structures at warehouses. This improvement is possible since the structure of the auxiliary graph allows additional arrival times to be added to nodes for commodities terminating at that node, without introducing any additional departure times (and subsequently introducing additional variables).}

Our main contributions are as follows:
\begin{enumerate}[noitemsep]
\item We develop an arc-based DDD approach that decreases the size of the formulation solved in each iteration for many network structures, including hub-and-spoke networks.
\item We prove that the arc-based DDD approach generalizes the node-based approach (Theorem \ref{theorem:arc_based_equiv}). Specifically, given any initial partial network and refinement process defining a node-based DDD algorithm, there is a corresponding initial partial (auxiliary) network and refinement process for the arc-based DDD approach which leads to the same algorithm.
\item We implement an arc-based DDD algorithm for the problem of SND-RR and apply the algorithm to families of instances based on the original construction outlined in \cite{Boland1}. We consider three sets of instances:
\begin{enumerate}
	\item SND-RR with designated paths: each commodity has a designated path in the flat (physical) network;
	\item SND on hub-and-spoke networks: the flat network has a hub-and-spoke structure and commodity routes are unrestricted;
	\item {\color{black} SND with critical times: the flat network is randomly chosen as in \cite{Boland1}, commodity routes are unrestricted, and each node has a set of critical times (release times and deadlines). }\\	
\end{enumerate}
Our computational results are summarized in Table \ref{tab:comp_summary}, where the variable and constraint counts are with respect to the lower bound formulation in the final iteration of each DDD algorithm. Each entry indicates the average percentage reduction in that field when solving an instance with the arc-based DDD approach compared to solving the instance with the node-based DDD approach.

\begin{table}[h!]
\centering
\resizebox{.75\textwidth}{!}{\begin{tabular}{l|c|c|c}
\toprule
Family of instances &  runtime & variable count & constraint count \\
\hline 
SND-RR with designated paths & 57\% & 45\% & 40\% \\
SND on hub-and-spoke networks & 42\% & 34\% & 31\% \\
SND with critical times & 12\% & 20\% & 19\% \\
\bottomrule
\end{tabular}}%
\caption{Average $\%$ decrease for arc-based DDD  over  node-based DDD.}
\label{tab:comp_summary}
\end{table}
\FloatBarrier
\end{enumerate}

The implementation of the arc-based and node-based DDD approaches as well as the benchmark instances can be found at \href{https://github.com/madisonvandyk/SND-RR}{\textsf{https://github.com/madisonvandyk/SND-RR}}.

\section{Related work}\label{sec:lit_rev}
Temporal network design problems were first introduced by Ford and Fulkerson \cite{FordFulkerson58, FordFulkerson59} in the context of network flow theory. The authors demonstrated that these problems can be reframed as static network flow problems in the corresponding time-expanded network. In the case of multicommodity flows, Hall et al.~\cite{Hall2003} provide hardness proofs as well as polytime solvable instances.
Fleischer et al.~\cite{Fleischer2003, FleischerSkutella2007} provide guarantees on the cost increase of the optimal solution when we allow a coarser network and only include node copies for every $\Delta$ units of time. In this $\Delta$-condensed approach, each node shares the same set of departure times $\{0, \Delta, 2\Delta, \ldots, T\}$, as opposed to a partially time-expanded network in DDD where the departure times available to each node is non-uniform. Wang and Regan \cite{WangRegan2} show that iteratively refining a time window discretization for TSP-TW will converge to an optimal solution. Similarly, Dash et al.~\cite{Dash} iteratively refine a set of time periods based on a preprocessing scheme in contrast to the dynamic scheme in DDD. 

The DDD framework was first introduced and applied to the service network design problem by Boland et al.~\cite{Boland1}. Initial applications of the DDD paradigm addressed connectivity problems such as the shortest path problem \cite{DDDpaths} and the traveling salesman problem with time windows \cite{DDD_TSP}. Subsequently, the approach was extended to handle a variety of problems with varying costs and constraints \cite{flexible, LagosDDD, Pottel, VanDyk, Vu_hewitt}. For a complete presentation of the DDD framework for connectivity problems, we refer the reader to the survey of Boland and Savelsbergh \cite{DDD_survey}. 

Strategies to improve the DDD algorithm focus on limiting the number of iterations until termination, as well as speeding up each individual iteration. Marshall et al.~\cite{Interval_DDD} introduce the \emph{interval-based dynamic discretization discovery algorithm} (DDDI) which was demonstrated to find solutions to instances of SND faster than traditional DDD.  In DDDI, the time discretization in each iteration is interpreted as a set time intervals at each node, rather than the set of departure times at each node. This interpretation allows for a more effective refinement strategy so that in each iteration, the current optimal flat solution in the partial network (set of physical paths and consolidation of commodities onto arcs) is infeasible in the subsequent iteration. Scherr et al.~\cite{DDDAuto} suggest removing timed nodes if they are no longer required for a high quality solution. However, no removal strategy has been implemented or studied in the literature. Hewitt \cite{enhanced} explores speed-up techniques for DDD which include enhancements such as a two-phase implementation of DDD and the addition of valid cuts to strengthen the relaxed model in each iteration. 

Despite noteworthy improvements to the DDD paradigm, all of the previously proposed algorithm enhancements maintain the structure that all arcs departing a common node share the same set of departure times in each iteration. This is true even in the case of DDDI, where all arcs departing a common node share the same set of time intervals. While in DDDI each commodity has its own partial network, these networks are constructed by removing timed arcs that could not be used by that commodity in any solution and so the issue remains. Thus, even with these enhancements, the DDD framework is ineffective for problems where all near-optimal solutions require a large number of departure times at the majority of the nodes in the network. For this reason, our work tackles a previously unaddressed issue. Since our arc-based approach is orthogonal to the previously proposed improvements to DDD, our experimental results compare to the original DDD implementation for the service network design (SND) problem presented in \cite{Boland1}. We note that the same arc-based approach naturally extends to the setting of DDDI. 

SND problems frequently arise in fulfillment planning \cite{Crainic2, wieberneit} and involve determining paths in space for shipments, as well as scheduling their departure time on each leg. Boland et al.~\cite{Boland1} first applied the DDD paradigm to instances of SND, and subsequently many enhancements to the framework have focused on variants of this problem \cite{enhanced, flexible, Interval_DDD}. In practical instances of SND, there are additional operational realities impacting the set of feasible solutions for the physical networks. We define the \emph{service network design problem with restricted routes} (SND-RR) as the problem of SND where the physical path of each commodity must be contained in a predefined subgraph of the flat network. Hewitt \cite{enhanced} considers the problem of SND where each shipment must follow a path chosen from a predefined set of candidate paths. Similarly, Lara et al.~\cite{amazon} solve large fixed-charge problems where each shipment has a \emph{single} candidate path. 

Operational realities are often expressed by the structure of the underlying physical network. For instance, many low-cost and efficient networks follow a region-based structure. In particular, a \emph{hub-and-spoke network} is a common region-based network structure used in airline and fulfillment networks including the distribution networks of FedEx and UPS \cite{bowen}. 
The prevalence of these region-based networks highlights the importance of designing algorithms that are effective for these particular graph structures. Recently, Lara et al.~\cite{amazon} gave a Lagrangian decomposition-based heuristic approach for their temporal fixed charge problem where the decomposition roughly follows regional lines. Our arc-based DDD approach offers an alternative exact algorithm for solving temporal network design problems with region-based structures. 

\subsection{Roadmap}
This paper is organized as follows. In Section \ref{sec:background}, we define the problem of SND-RR, and describe a DDD algorithm for this problem that is analogous to the original implementation for SND in \cite{Boland1}. In Section \ref{sec:auxiliary} we define an auxiliary network based on a partitioning of the arc set in the flat network, and we present a time-indexed formulation for SND-RR on this auxiliary network. In Section \ref{sec:arc-based}, we define a DDD algorithm based on this new time-indexed formulation and prove its correctness. We also prove that this arc-based DDD approach generalizes the original node-based approach. Specifically, given any node-based implementation of DDD, we can select an initial partial (auxiliary) network and refinement process for the arc-based approach to obtain the same algorithm. In Section \ref{sec:applications} we outline the potential for the arc-based DDD approach to improve over the node-based approach for region-based networks, and when paths are designated in the base graph for each commodity. Finally, in Section \ref{sec:computational_results} we present computational results demonstrating the improvement of our arc-based discretization over the original node-based approach.

\section{Problem statement and background}\label{sec:background}
In this section we define time-indexed formulations and the problem of service network design with restricted routes. We then provide an overview of a DDD algorithm for solving SND-RR which is analogous to that of the original DDD algorithm for SND. This original DDD implementation serves as the baseline for our experiments in Section \ref{sec:computational_results} since our strategy to reduce DDD iteration sizes via an arc-based approach is orthogonal to the previously proposed enhancements to DDD.

\subsection{Time-expanded networks}
There are two main model types used to solve temporal network design problems: continuous formulations and time-indexed formulations. \emph{Continuous formulations} use continuous variables to model timing decisions\footnote{Instead of having binary variables for each commodity, time, and arc triplet, continuous formulations have a variable for each commodity and arc pair, and its value is equal to the time the commodity traverses the arc.}, whereas \emph{time-indexed formulations} have variables and constraints indexed by each time point a decision could be made \cite{LagosDDD}. A  time-indexed formulation for a (discrete-time) temporal problem is obtained by expressing the temporal problem as a static problem in the corresponding \emph{time-expanded network}. 

Let $D = (N, A)$ be a directed graph with nodes $N$ and arcs $A$, along with arc transit times $\tau$. Let $T$ be the time horizon under consideration. The corresponding time-expanded network, $D_T = (N_T, A_T \cup H_T)$, consists of a copy of each node $v \in N$ for each time point $t \in [T]:=\{0,1, \ldots, T\}$. The \emph{movement} arcs, $A_T$, consist of a copy of each arc in $A$ for each departure time, and the \emph{storage} arcs, $H_T$, connect each node copy to the next consecutive copy. Specifically, 
\begin{align*}
	N_T & := \{(v,t): v \in N, t \in [T]\},\\
	A_T & := \{((v,t), (w, t+ \tau_{vw})): (v,t) \in N_T, 			vw \in A, t+ \tau_{vw} \leq T\},\mbox{ and }\\
	H_T & := \{((v,t), (v, t+1): v \in V, t \in [T-1]\}.
\end{align*}
An example is presented in Figure \ref{fig:full_diagram_D_T}. We will often refer to arcs and nodes in a time-expanded network as \emph{timed arcs} and \emph{timed nodes} respectively. We will write timed nodes with their associated times as $(v,t)$, and the corresponding node in the base graph is $v$. Similarly, we write timed arcs as $((v,t), (w,t'))$, and the corresponding arc in the base graph is $vw$. Additionally, we will refer to paths in the time-expanded network as \emph{trajectories}. We will often refer to paths in the $D$ as \emph{flat paths}. 

\begin{figure}[!htb]
\centering
\includegraphics[width=1\textwidth]{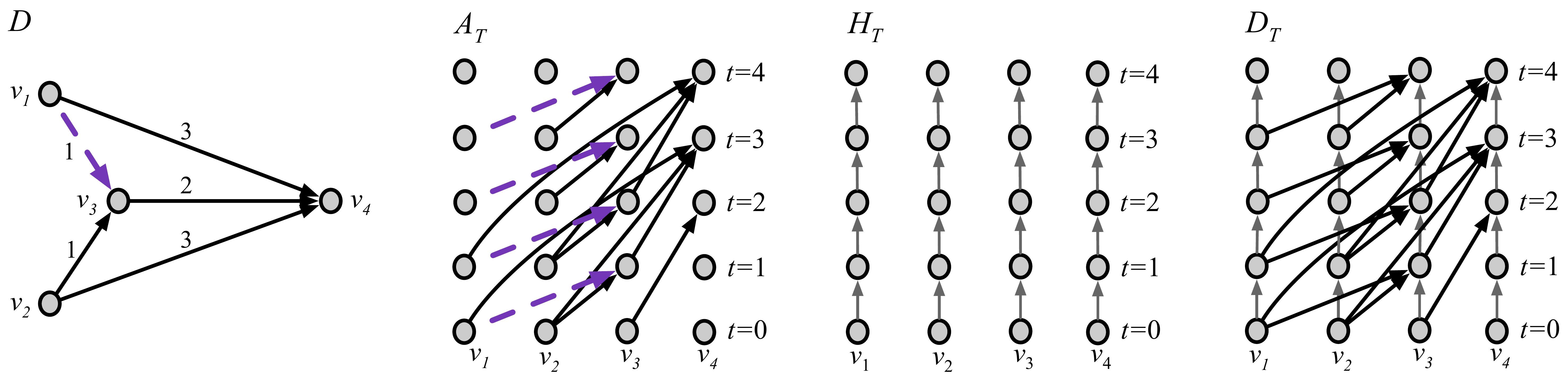}
\caption{Base graph $D$ and the construction of $D_T$ when $T = 4$.}
\label{fig:full_diagram_D_T}
\end{figure}
\FloatBarrier

A temporal network design problem with time horizon $T$ can be modeled as a static network design problem in $D_T$ after assigning appropriate costs, capacities, and demands. We present one such model for the service network design problem in the following section.

%%%%%%%%%%%%%%%%%%%%%%%%%%%%%%%%%%%%%%%%
%%%%%%%%%%%% SND definition %%%%%%%%%%%%
%%%%%%%%%%%%%%%%%%%%%%%%%%%%%%%%%%%%%%%%

\subsection{Service Network Design}
In an instance of \emph{service network design} (SND), we are given a directed graph $D = (N,A)$ with node set $N$ and arc set $A$, and a set of commodities $\mathcal{K}$. We refer to $D$ as the \emph{flat} network. Each arc $vw \in A$ has an associated transit time $\tau_{vw} \in \mathbb{N}_{> 0}$, a commodity-dependent per-unit-of-flow cost $c^k_{vw} \in \mathbb{R}_{> 0}$ for each $k \in \K$, a fixed cost $f_{vw} \in \mathbb{R}_{> 0}$, and a capacity $u_{vw} \in \mathbb{N}_{> 0}$. We interpret $u_{vw}$ as the capacity of trucks scheduled on arc $vw$, and $f_{vw}$ as the per-truck-cost on $vw$. Each commodity $k \in \mathcal{K}$ has a source $o_k \in N$ and sink $d_k \in N$, along with a demand $q_k$ that must be routed along a single trajectory from $o_k$ to $d_k$ (i.e. it is unsplittable). Let $r_k$ denote the time commodity $k$ becomes available at its origin (its release time) and let $l_k$ denote the deadline for commodity $k$. Let $T = \max_{k \in \K} l_k$ be the latest deadline among the commodities. We may assume that the earliest release time is time 0 by shifting the time horizon. An $o_k, d_k$-trajectory in $D_T$ is \emph{feasible} for commodity $k$ if it departs $o_k$ no earlier than $r_k$ and arrives at $d_k$ no later than $l_k$. The goal of SND is to determine a feasible trajectory in $D_T$ for each commodity in order to minimize the total fixed and variable cost of the resources required to ship the commodities along the trajectories. 

We define \emph{service network design with restricted routes} (SND-RR) as the problem of SND with the additional input of a designated subgraph $D^k \subseteq D$ for each commodity $k \in \K$. In SND-RR,  an $o_k, d_k$-trajectory $Q$ in $D_T$ is \emph{feasible} for commodity $k$ if it departs $o_k$ no earlier than $r_k$, arrives at $d_k$ no later than $l_k$,  and the flat path corresponding to the trajectory $Q$ is fully contained in $D^k$. 

\subsubsection*{Time-indexed formulation for SND-RR}
We now define a time-indexed formulation for SND-RR analogous to the formulation presented by Boland et al.~in \cite{Boland1}. For each $k \in \K$, let $D^k_T = (N^k_T, A_T^k \cup H_T^k)$ denote the time-expanded network for $D^k$ with time horizon $T$. For each commodity $k \in \K$ and each timed arc $a \in A_T^k \cup H_T^k$, let $x_a^k$ be a binary variable which is equal to 1 if commodity $k$ is scheduled to travel along timed arc $a$ in its assigned trajectory. For each $a \in A_T$, let $y_a$ denote the number of trucks scheduled along timed arc $a$. For each timed arc $a = ((v,t),(w,t')) \in A_T$, let $u_a := u_{vw}$, $c_a := c_{uv}$, and $f_a := f_{uv}$. Let $\delta^+_{D_T} (v,t)$ and $\delta^-_{D_T} (v,t)$ denote the outgoing and incoming timed arcs at $(v,t)$ in $D_T$. That is, $\delta^+_{D_T} (v,t) = \{a \in A_T \cup H_T: a = ((v,t),(w,t'))\}$, and $\delta^-_{D_T} (v,t) = \{a \in A_T \cup H_T: a = ((w,t'),(v,t))\}$. 

The following integer program (IP) models SND-RR. 

\begin{align}\tag{SND-RR($D_T$)}\label{IP:D_T}
    \min~~ & \sum_{a \in A_T} f_a y_a + \sum_{a \in A_T} \sum_{k \in \K} c^k_a  q_k x_a^k \\
    \vspace{.5cm}
    \mbox{s.t.} ~~
	& x^k(\delta_{D_T^k}^{+}(v,t)) - x^k(\delta_{D_T^k}^{-}(v,t)) = 	\begin{cases}
        1 ~(v,t) = (o_k, r_k) \\
        -1 ~(v,t) = (d_k, l_k) \\
        0 ~\mbox{otherwise} \\
    \end{cases} \quad \forall k \in \mathcal{K}, (v,t) \in N^k_T 		\label{const1:D_Tflow}\\    
	\vspace{.5cm}
    & \sum_{k \in \mathcal{K}} q_k x_a^k \leq u_a y_a \quad \forall a \in A_T \label{const1:D_Tcapacity}\\
    & x_a^k \in \{0,1\} \quad \forall k \in \mathcal{K}, a \in A_T^k \cup H_T^k \label{var1:flow} \\
	& y_a \in \mathbb{N}_{\geq 0} \quad \forall k \in \mathcal{K}, a \in A_T \label{var2:flow}
\end{align}
Constraint ($\ref{const1:D_Tflow}$) ensures that each commodity is assigned a feasible trajectory. Constraint ($\ref{const1:D_Tcapacity}$) ensures that for each timed arc there is sufficient capacity purchased to accommodate the trajectories scheduled to traverse that timed arc. The objective function optimizes for the total fixed and variable cost. Note that when $D^k = D$ for all $k \in \K$, this is the same as the timed-indexed formulation for SND in \cite{Boland1}.

%%%%%%%%%%%%%%%%%%%%%%%%%%%%%%%%%%%%%%%%
%%%%%%%%%%% DDD Application %%%%%%%%%%%%
%%%%%%%%%%%%%%%%%%%%%%%%%%%%%%%%%%%%%%%%

\subsection{Node-based DDD framework}
The downside to modeling temporal problems with time-expanded networks is that this time-expansion often makes the resulting static problem impractical to solve since the network grows linearly in $T$ \cite{FordFulkerson58, FordFulkerson59}. In an effort to avoid solving these large time-indexed formulations, a DDD algorithm solves a series of MIPs defined on smaller \emph{partially time-expanded networks} that appropriately \emph{underestimate} the transit times of the arcs in $A_T$. A partially time-expanded network with respect to $D$ and $T$ is any directed graph $D_S = (N_S, A_S \cup H_S)$ where $N_S \subseteq N_T$ and $A_S \subseteq \{((v,t), (w, t')): (v,t) \in N_T, vw \in A, t' \leq t+ \tau_{vw} \leq T\}$, and $H_S$ connects each node copy to its next copy in $N_S$. We will refer to fully time-expanded networks and partially time-expanded networks as \emph{full} and \emph{partial} networks respectively. In this paper, $D_T = (N_T, A_T \cup H_T)$ denotes the full network with time horizon $T$, and $D_S = (N_S, A_S \cup H_S)$ is any partial network. 

In a DDD algorithm, the partial networks and corresponding MIPs are constructed to ensure that they provide a lower bound on the optimal value of the original problem. In the context of SND-RR, we denote the lower bound formulation corresponding to a partial network $D_S$ as SND-RR($D_S$). We define SND-RR($D_S$) so that it is the formulation corresponding to routing the commodities in $D_S$ rather than $D_T$, with the same induced costs. Since this formulation is analogous to the formulation presented in \cite{Boland1}, we defer the details to Appendix \ref{app:lower_bound}. If a solution to the partial network cannot be converted to a solution to the original instance of equal cost, the partial network is refined by adding timed nodes and timed arcs based on the current solution. An overview of this method is presented in Algorithm \ref{alg:ddd_outline1}. 

\begin{algorithm}
	\DontPrintSemicolon
	\KwIn{Base network $D = (N,A)$, commodity set $\mathcal{K}$ with subgraph $D^k$ for each $k \in \K$}
	\textbf{Initialization}: $D_S \leftarrow D_0$, where SND-RR($D_S$) provides a lower bound for SND-RR($D_T$)\\
	\While{not solved}{
	    \textbf{Lower bound}: Solve SND-RR($D_S$) and obtain a solution $(\hat{x}, \hat{y})$ in $D_S$ \\
        \textbf{Upper bound/termination}: determine if $\hat{x}$ can be converted to a solution to SND-RR($D_T$) with equal cost \\
        \If{yes}{
            Stop. An optimal solution has been found for SND-RR($D_T$).}
        \textbf{Refinement}: update $D_S$ while ensuring SND-RR($D_S$) provides a lower bound for SND-RR($D_T$).
	}
	\caption{$\mathtt{SND\mbox{-}RR}(D, \mathcal{K}$)}
	\label{alg:ddd_outline1}
\end{algorithm}	
\FloatBarrier

We now describe the specific initialization, lower and upper bound, and refinement steps in the original DDD implementation for SND presented in \cite{Boland1}. The only change is that instead of solving SND($D_S$) in each iteration, we solve SND-RR($D_S$). While this is a minimal change, we provide a proof of correctness for completeness in Appendix \ref{app:lower_bound}. Note, SND($D_S$) is equivalent to SND-RR($D_S$) when $D^k = D$ for all $k \in \K$. We will refer to this DDD algorithm as the \emph{node-based} DDD approach. 

\subsubsection{Initialization and lower bound}
When $D^k = D$ for all $k \in \K$, Boland et al.~prove that so long as the partial network $D_S$ satisfies the following two properties, an optimal solution of SND-RR($D_S$) provides a lower bound on the optimal value of SND-RR($D_T$) (Theorem 2 in \cite{Boland1}).

\begin{itemize}[noitemsep]
            \item[(P1)]\textbf{Timed nodes}: For all $k \in \mathcal{K}$, $(o_k, r_k)$ and $(d_k, l_k)$ are in $N_S$;

\hspace{2.42cm} For all $v \in N$, $(v,0)$, and $(v,T)$ are in $N_S$;
            \item[(P2)]\textbf{Arc copies}: For all $vw \in A$, for all $(v,t) \in N_S$ with $t + \tau_{vw} \leq T$, we have
            $((v,t), (w, t')) \in A_S$ where \[t' = \max \{r: r \leq t + \tau_{vw}, (w, r) \in N_S\}.\]
        \end{itemize}
Note that property (P2) gives a natural method to construct $D_S$ given $N_S$. In its standard implementation, in the first iteration of a DDD algorithm the partial network is initialized to be the network consisting only of the timed nodes stated in (P1). Throughout this paper we will assume that any partial network has no additional arcs beyond those stated in (P2). Thus, the partial networks we work with are completely determined by the flat network and the set of timed nodes. 

The following generalization of Theorem 2 in \cite{Boland1} follows directly from a restatement of the proof and the observation that physical routes of each commodity are unchanged. We provide a sketch of this proof and defer the formal proof to Appendix \ref{app:lower_bound}.

\begin{theorem}\label{theorem:LB}
When the partial network $D_S$ satisfies properties (P1) and (P2), an optimal solution of SND-RR($D_S$) provides a lower bound on the optimal value of SND-RR($D_T$).
\end{theorem}

\begin{proof}[Proof sketch]
Let $(\bar{x}, \bar{y})$ be a feasible solution to SND-RR($D_T$) with cost $C$, and let $\mathcal{Q} = \{Q_k\}_{k \in \K}$ denote the corresponding set of trajectories. Let $\mu: A_T \rightarrow A_S$ be the map defined so that for each timed arc $a \in A_T$, 
\begin{equation}\label{mu_eq_D_S}
    a = ((v, t), (w,t')) \quad \rightarrow  \quad \mu(a) = ((v,\hat{t}), (w, \hat{t}')), 
\end{equation}  where $\hat{t} = \max\{s: s \leq t, (v, s) \in N_S\}$, and $\hat{t}' = \max\{s: s \leq \hat{t} + \tau_{vw}, (v, s) \in N_S\}$. Observe that $\hat{t}'$ is dependent on $\hat{t}$ rather than $t'$. We obtain trajectories $\hat{\Q} = \{\hat{Q}_k\}_{k \in \K}$ by mapping each movement arc $a \in A_T$ to $\mu(a) \in A_S$, and forming trajectories by adding holdover arcs. Due to properties (P1) and (P2), the resulting trajectories are feasible in $D_S$. The underlying paths in the flat network for $\hat{\Q}$ and $\Q$ are the same for each commodity, and whenever two trajectories share the same timed arc $\Q$, the two trajectories share the same timed arc in $\hat{\Q}$. As a result, the cost of $\hat{\Q}$ is at most $C$. 
\end{proof}

%%%%%%%%%%%%%%%%%%%%%%%%%%%%%%%%%%%%%%%%
%%%%%% Upper bound and refinement %%%%%%
%%%%%%%%%%%%%%%%%%%%%%%%%%%%%%%%%%%%%%%%

\subsubsection{Upper bound and refinement} 

Let $(\hat{x}, \hat{y})$ be an optimal solution to SND-RR($D_S$) with corresponding trajectories $\Q = \{Q_k\}_{k \in \K}$ in $D_S$, and paths $\P = \{P_k\}_{k\in\K}$ in $D$. We obtain a feasible solution to \ref{IP:D_T} with the same flat paths $\P$ (and hence same variable cost as $\Q$) by specifying a departure time for each arc in $P_k$, for each commodity $k \in \K$, that satisfies the release time and deadline of $k$ and respects the actual transit time of the arcs. In addition, if the set of departure times is constructed so that every pair of commodities dispatched together on an arc $vw$ in $\Q$ still shares the same departure time for arc $vw$, the corresponding feasible solution has the same fixed cost as $\Q$. Thus, the paths $\P$ along with a set of departure times that satisfies the aforementioned feasibility and consolidation requirements provides an optimal solution to \ref{IP:D_T}. 

In order to ensure there is a set of departure times that is feasible for the set of paths $\P$ proposed in each iteration (satisfies release times, deadlines, and actual transit times), Boland et al.~add the following constraint to SND-RR($D_S$):
\begin{align}
\sum_{a \in A_S} \tau_a x_a^k \leq l_k - r_k, \quad \forall k \in \K, \label{const:feasible_D_T}
\end{align}
where for $a = ((v,t), (w,t')) \in A_S$, $\tau_a$ denotes the transit time of arc $vw$ rather than $t' - t$. Constraint (\ref{const:feasible_D_T}) ensures that for each trajectory in $Q$ there is a feasible trajectory in $D_T$ with the same underlying path in the flat network.

The continuous formulation presented in \cite{Boland1} has a variable $t^k_{v}$ for each commodity $k \in \K$ and each non-destination node $v$ in $P_k$ that models the departure time for $k$ on the arc leaving $v$ in $P_k$. The feasibility of the departure times is captured in the following constraints, where for fixed $k \in \K$, $v_p$ is the $p$th node in $P_k$. 
\begin{align}
t_{v_p}^k + \tau_{v_p v_{p+1}} & \leq t_{v_{p+1}}^k \quad \forall k \in \K, p \in [|P_k| - 1] \label{const:cont1A}\\
r_k & \leq t_{o_k}^k \quad \forall k \in \K \label{const:cont2A}\\
t_{v_{|P_k| - 1}} + \tau_{v_{|P_k| - 1} d_k} & \leq l_k \quad \forall k \in \K, \label{const:cont3A}
\end{align} 

For each arc $vw \in A$, let $J_{vw}$ denote the set of pairs of commodities that traverse arc $vw$ at the same time in $\Q$ (i.e. use the same timed copy of $e$ in $\Q$), and let $\J := \{J_{vw}\}_{vw \in A}$. Boland et al.~introduce a variable $\delta_{vw}^{k_1 k_2}$ for each pair of commodities $(k_1, k_2) \in J_{vw}$, for each arc $vw \in A$, along with the following set of constraints.
\begin{align}
\delta_{vw}^{k_1 k_2} & \geq t_v^{k_1} - t_v^{k_2} \quad \forall (k_1, k_2) \in J_{vw}, \forall {vw} \in A \label{const:cont5A}\\
\delta_{vw}^{k_1 k_2} & \geq t_v^{k_2} - t_v^{k_1} \quad \forall (k_1, k_2) \in J_{vw}, \forall {vw} \in A.\label{const:cont6A}
\end{align}
Observe that $\delta_{{vw}}^{k_1 k_2}$ must always be non-negative, and $\delta_{{vw}}^{k_1 k_2}$ is zero precisely when the departure times for commodities $k_1$ and $k_2$ on arc ${vw}$ are the same. Thus, if there is a vector $(\bar{t}, \bar{\delta})$ that satisfies constraints (\ref{const:cont1A})-(\ref{const:cont6A}) with $\bar{\delta} = \mathbf{0}$, then we have found an optimal solution to \ref{IP:D_T}. If not, then the current partial network, $D_S$, must be modified. 

To update the partial network, we want to add a small number of timed nodes that correct the current set of infeasible trajectories. Boland et al.~stipulate that in the continuous formulation, the departure times for commodities with feasible trajectories in $\Q$ (those that consisted only of timed arcs with realistic transit times), are unchanged. Let $\K^F \subseteq \K$ denote the set of commodities with feasible trajectories in $\Q$, and let $\T^F = \{\T_k\}_{k \in \K^F}$ where $\T_k = \{\bar{t}_{v_p}^k\}_{p \in [|P_k|-1]}$ is the set of departure times for the arcs in $P_k$ given by $\Q$. The following constraints ensure the departure times for the set $\K^F$ are unchanged. 
\begin{align}
t_{v_p}^k & = \bar{t}_{v_{p}}^k, \quad \forall k \in \K^F, p \in [|P_k| - 1] \label{const:cont4A}
\end{align}

In total, the continuous formulation presented in \cite{Boland1} is as follows. 

\begin{align}\tag{LP-UB$(\P, \J, \K^F, \T^F)$}\label{LP:UB}
 \min~~ & \sum_{{vw} \in A} \sum_{(k_1, k_2) \in J_{vw}} \delta_{{vw}}^{k_1 k_2} \\
    \vspace{.5cm}
    \mbox{s.t.} ~~
	& (\ref{const:cont1A}) - (\ref{const:cont4A})
\end{align}
Due to the addition of constraint (\ref{const:feasible_D_T}), in each iteration there is a feasible solution $(\bar{t}, \bar{\delta})$ to the continuous formulation \ref{LP:UB}, which corresponds to a feasible solution $(\bar{x}, \bar{y})$ to \ref{IP:D_T}. If $\bar{\delta} = \mathbf{0}$, then $(\bar{x}, \bar{y})$ is an optimal solution to \ref{IP:D_T}. Otherwise, a non-empty subset of the $\delta$-variables have positive value. For each such variable $\delta_{vw}^{k_1 k_2}$ with positive value, one of $k_1$ and $k_2$ must a trajectory in $D_S$ which has a short timed arc by construction of $\K^F$ and the addition of constraint (\ref{const:cont4A}). Let $\C := \{k_1 \in \K \setminus \K^F: \exists {vw} \in A, k_2 \in \K, \bar{\delta}_{vw}^{k_1 k_2} > 0\}$ capture this set of commodities. We see that $\C$ is a non-empty, and each commodity in $\C$ has a trajectory in $D_S$ defined by $\hat{x}$ with a short timed arc. In the refinement step, for each $k \in \C$ a timed node is added to lengthen the short timed arc with the earliest departure time in the trajectory $Q_k$.

\section{An auxiliary graph and new formulation for SND-RR}\label{sec:auxiliary}

The current construction of partial networks makes DDD ineffective for solving problems that require many departure times at a large number of nodes. For such problems, in order to express an optimal solution on a partial network $D_S = (N_S, A_S \cup H_S)$, $N_S$ must include a large portion of $N_T$. {\color{black} In each iteration, the partial network $D_S$ is constructed given the set $N_S$ according to property (P2):

\begin{enumerate}
\item[(P2)] For all $vw \in A$, for all $(v,t) \in N_S$ with $t + \tau_{vw} \leq T$, we have $((v,t), (w, t')) \in A_S$ where \[t' = \max \{r: r \leq t + \tau_{vw}, (w, r) \in N_S\}.\]
\end{enumerate}

Since this construction results in a timed copy of each arc in $vw \in A$ for each timed copy of $v$ in $N_S$, when $N_S$ is large, this forces $A_S$ to be large as well. As a result, solving the lower bound formulation in the final iteration would be comparable to solving the full time-indexed formulation in the first place.}

For example, consider the graph structure in Figure   \ref{fig:star_motivation} with $m$ arcs departing a node $v$, all with unit transit times. For each $i \in [m]$, suppose there is a commodity $k$ with origin $v$, destination $v_i$, and release time $i$. Due to the current structure of partial networks, the initial partial network $D_0 = (N_0, A_0 \cup H_0)$ would have a copy of node $v$ at all times in $[m]$. As a result, $D_0$ would have a copy of each arc for each departure time in $[m]$. Thus, we would have $A_0 \approx A_T$, and so the size of \ref{IP:D_S} is approximately the same as \ref{IP:D_T}. However, observe that an optimal solution can be expressed on a small subgraph of $D_T$: include a single copy of arc $v v_i$ departing $v$ at time $i$ for each $i \in [m]$. This simple instance highlights the issue with high-degree nodes in the network and the current construction of partial networks. Extensions of the original node-based DDD algorithm in \cite{enhanced, Interval_DDD} consist of refinement processes that can lead to similar problematic structures.

\begin{figure}[!htb]
\centering
\includegraphics[width=.2\textwidth]{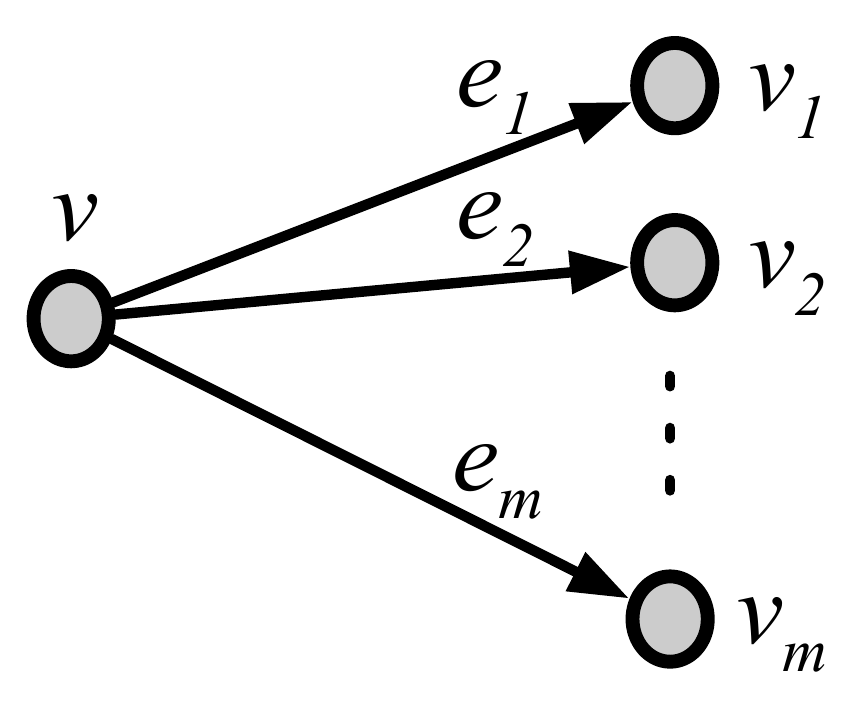}
\caption{}
\label{fig:star_motivation}
\vspace{-.25cm}
\end{figure}
\FloatBarrier

We would like to modify the family of partial networks and corresponding lower bound formulations so that permitting $d$ departure times at a node $v$ does not necessarily lead to $d$ timed copies of each arc departing $v$. In the current DDD approach, the partial network in each iteration is defined by the time discretization on the node set, denoted $\T_N = \{\T_v\}_{v \in N}$, where $\T_v \subseteq \{0, 1, \ldots, T\}$ is the set of departure times for all arcs with tail $v$. Instead, we would like each iteration to be defined by a discretization on the arc set, denoted $\T_A = \{\T_{uv}\}_{uv \in A}$, where $\T_{uv}\subseteq \{0, 1, \ldots, T\}$ is the set of departure times for arc $uv$. Furthermore, we would like this modification to still follow the main structure of the original DDD approach: 

\begin{enumerate}[noitemsep]
\item A time discretization, $\T$, dictates a set of departure times rather than a set of time intervals.
\item $\T$ and the base graph $D$ define a partial network $D_S = (N_S, A_S \cup H_S)$, and the corresponding lower bound formulation routes the commodities through $D_S$ rather than $D_T$.
\item As the discretization becomes finer the corresponding lower bound converges to the optimal value of \ref{IP:D_T}.
\item The partial network is uniquely defined by the current discretization and independent of previous iterations.
\end{enumerate}
One way to prove that a formulation (IP1) is a relaxation of another (IP2) is to exhibit a function that maps a feasible solution from IP2 to a feasible solution to IP1 which has no greater cost. In the DDD literature, the map used is $\mu$ as defined in Equation (\ref{mu_eq_D_S}), which takes each trajectory in $D_T$ and rounds down the departure time for each arc. Note that $\mu$ preserves the physical path of each trajectory, and we would like any new relaxation to have an analogous map that preserves all physical routes and relaxes time.

Consider the following instance, where the base graph $D$ is given in Figure \ref{fig:P2A}, and  a partial network $D_S$ satisfying property (P2) is given in Figure \ref{fig:P2B}. A natural first attempt to construct a partial network with arc-dependent departure times is to simply remove a subset of the timed arcs from the standard construction of $D_S$. In Figure \ref{fig:P2C}, we remove the timed arc $((v_2, 1), (v_4, 3))$. Observe that in the resulting graph, there is no feasible trajectory for the flat path $v_1, v_2, v_4$. The issue is that the copy of arc $v_1 v_2$ arrives at $v_2$ \emph{later} than the copy of $v_2 v_4$ departs $v_2$. Thus, there is no map from $D_T$ to $D_S$ for this construction that would allow us to preserve the physical path of each trajectory in $D_T$.
\vspace{-.2cm}
\begin{figure}[htbp]
    \centering
    \begin{subfigure}[b]{0.3\textwidth}
        \includegraphics[width=.9\textwidth]{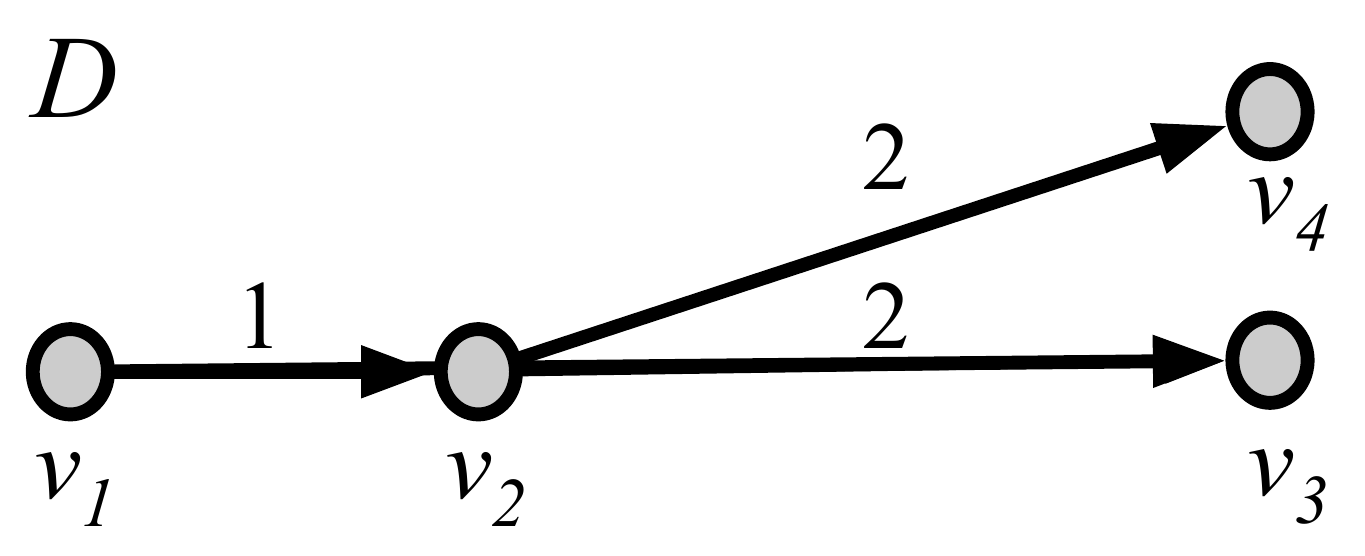}
        \caption{Base graph $D$}
        \label{fig:P2A}
    \end{subfigure}
    \begin{subfigure}[b]{0.3\textwidth}
        \includegraphics[width=.9\textwidth]{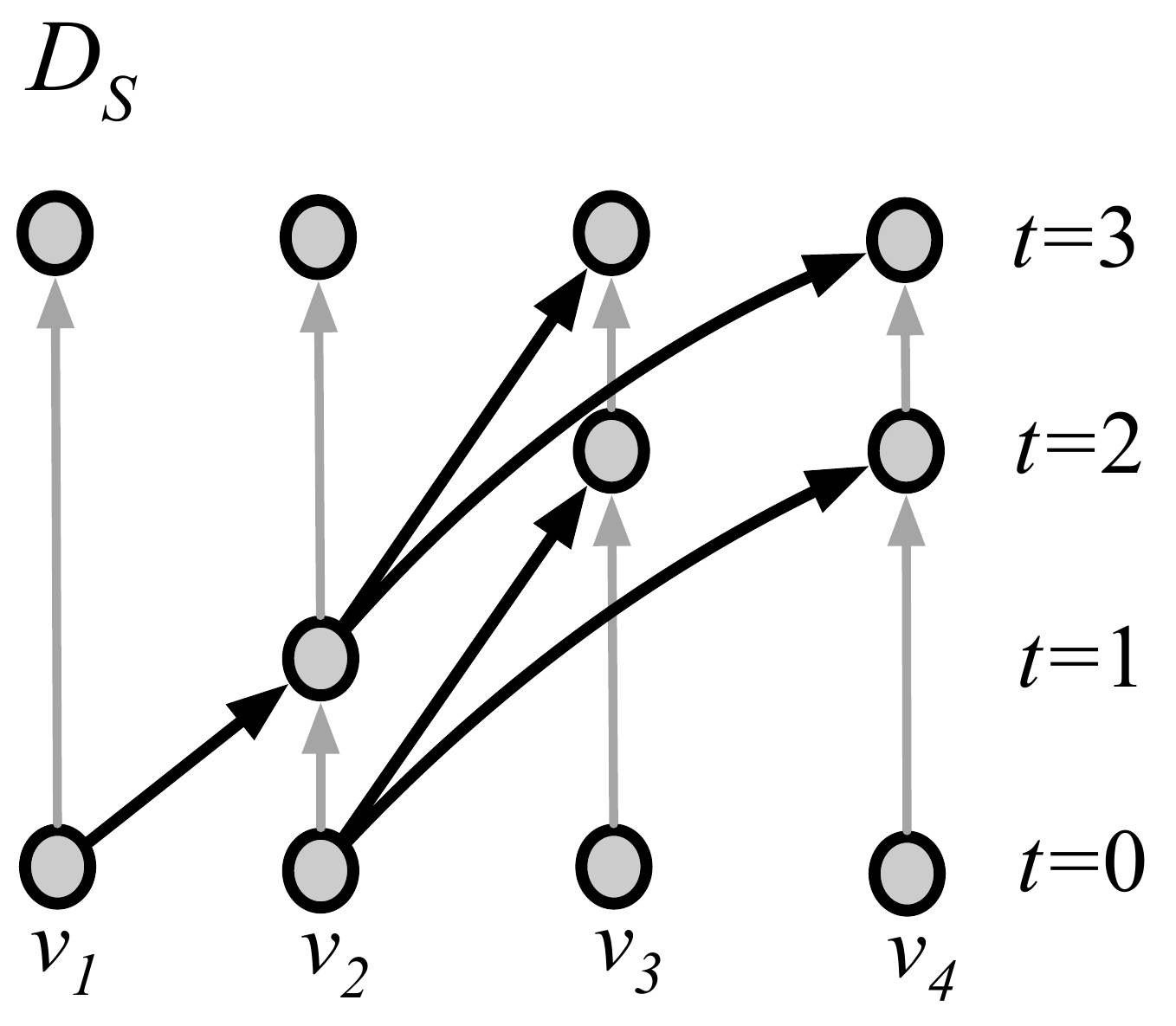}
        \caption{A partial network $D_S$.}
        \label{fig:P2B}
    \end{subfigure}
    \begin{subfigure}[b]{0.3\textwidth}
        \includegraphics[width=.9\textwidth]{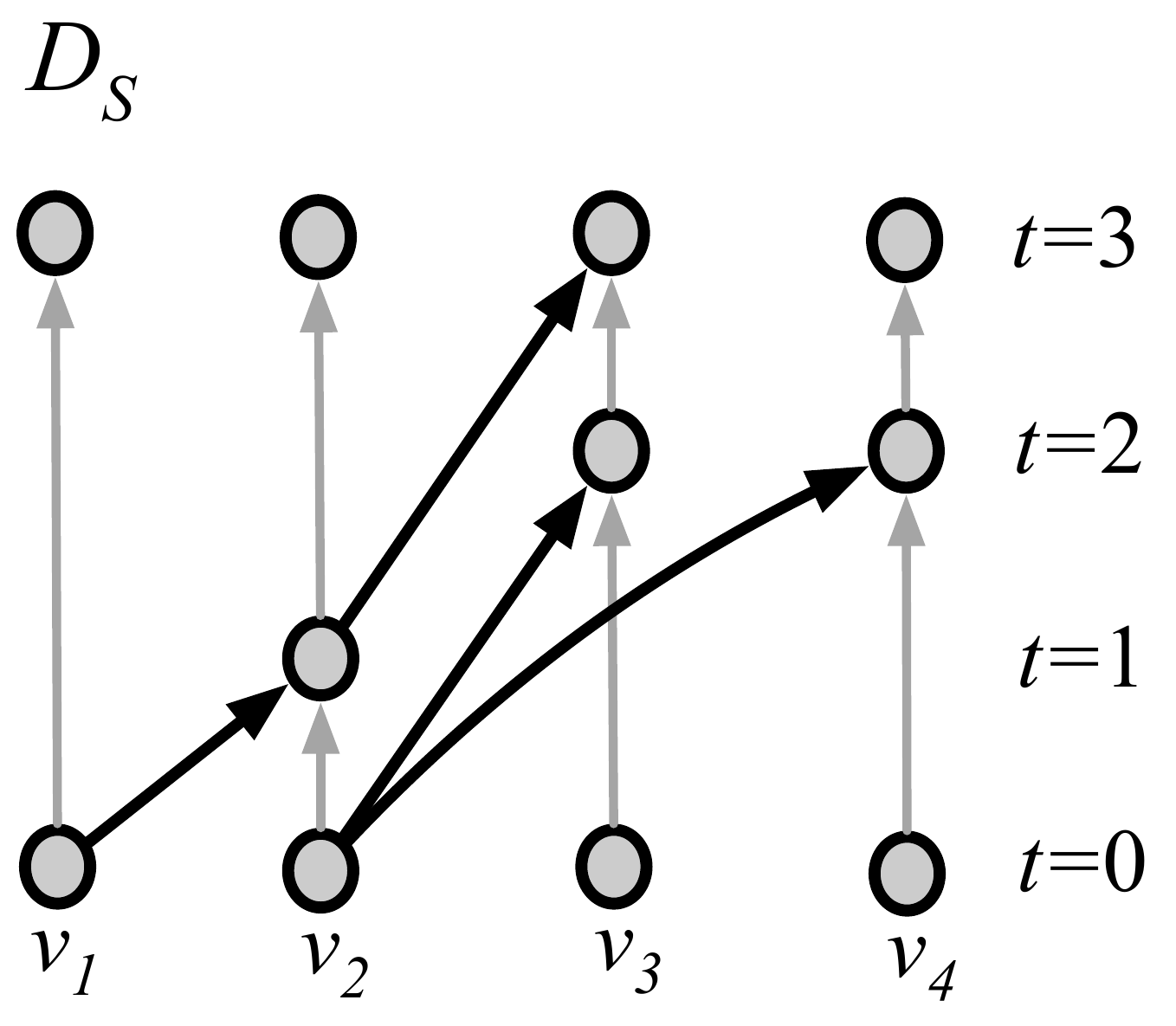}
        \caption{$D_S$ after removing a timed arc.}
        \label{fig:P2C}
    \end{subfigure}
    \caption{}
    \label{fig:P2}
\end{figure}
\FloatBarrier

If we replaced the timed arc {\color{myred}$a = ((v_1, 0), (v_2, 1))$} with {\color{myblue}$a' = ((v_1, 0), (v_2, 0))$} to form $D_S$, applying the standard map $\mu$ would certify that the optimal routing through $D_S$ provides a lower bound on the min-cost routing in $D_T$. However, this adjustment also leads to issues, as we demonstrate in the following example. In Figure \ref{fig:P2_description} we see that the timed arc $((v_2, 1), (v_3, 3))$ becomes redundant for most commodities: any trajectory using $((v_2, 1), (v_3, 3))$ (not originating at $v_2$) must arrive at $v_2$ at time 0, and so such a trajectory includes timed arcs $((v_2, 0), (v_2, 1))$ and $((v_2, 1), (v_3, 3))$. This pair of timed arcs can be replaced with $((v_2, 0), (v_3, 2))$ and $((v_3, 2), (v_3, 3))$, without incurring any additional cost. Thus, this arc shortening strategy may cause all arcs departing a fixed node to have the same set of non-redundant departure times. 

\begin{figure}[htbp]
    \centering
    \begin{subfigure}[b]{0.49\textwidth}
	\centering
        \includegraphics[width=.55\textwidth]{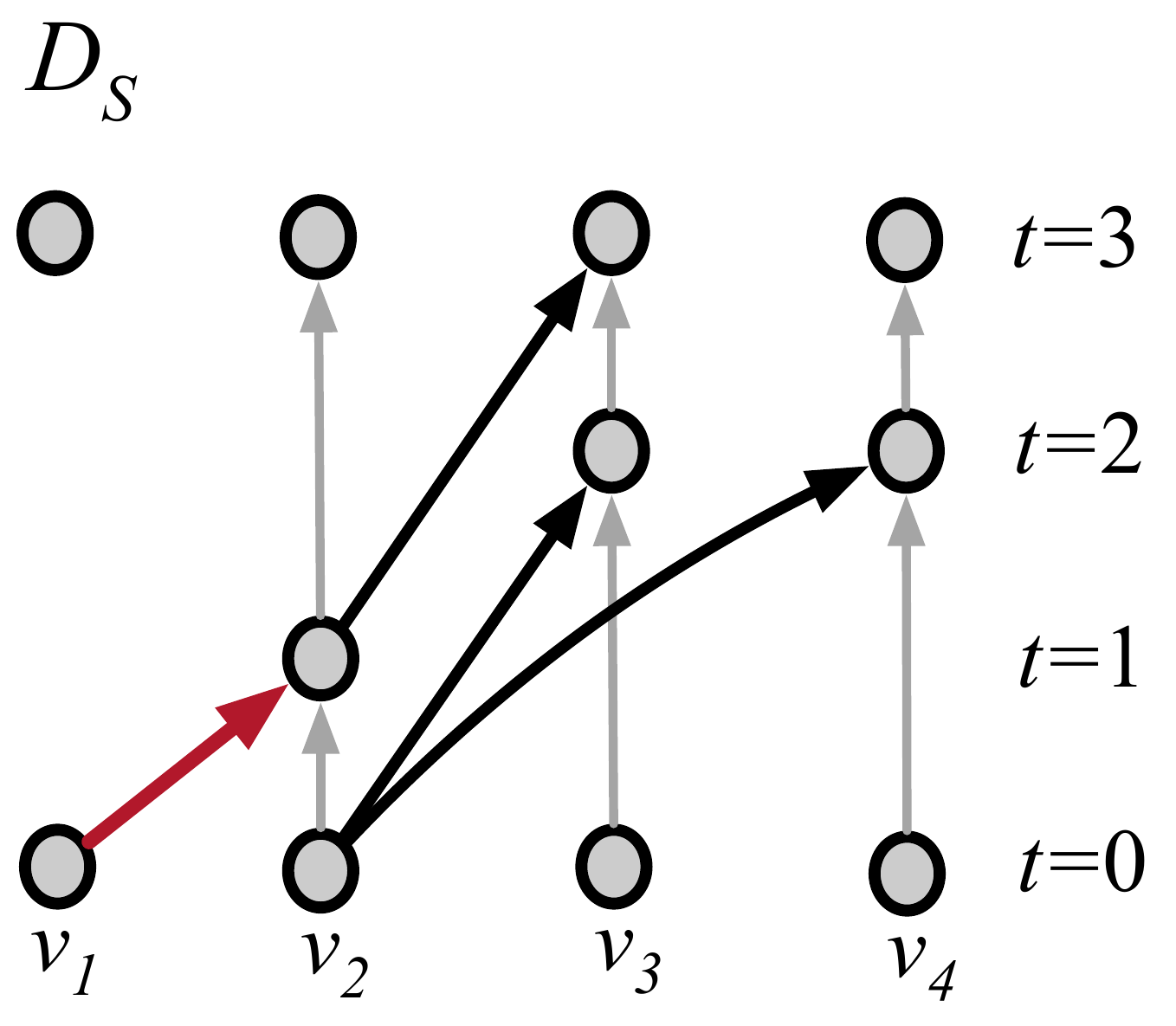}
        \caption{The $v_1v_2$ arc is {\color{myred}$a=((v_1, 0), (v_2, 1))$}.}
        \label{fig:P2_descriptionA}
    \end{subfigure}
    \begin{subfigure}[b]{0.49\textwidth}
	\centering
        \includegraphics[width=.55\textwidth]{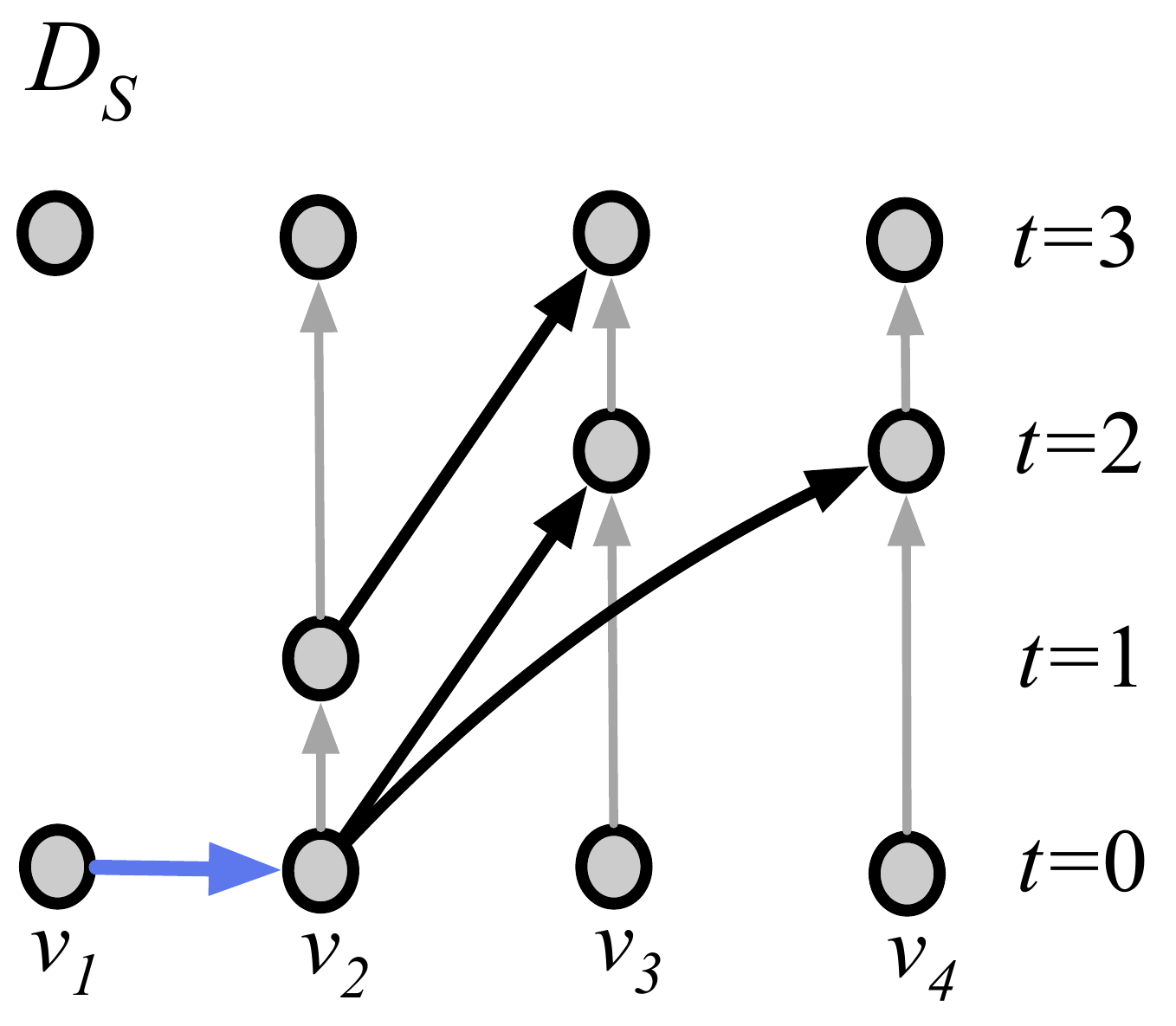}
        \caption{The $v_1v_2$ arc is {\color{myblue}$a' = ((v_1, 0), (v_2, 0))$}.}
        \label{fig:P2_descriptionB}
    \end{subfigure}
    \caption{}
    \label{fig:P2_description}
\end{figure}
\FloatBarrier

To mitigate the issue demonstrated in Figures \ref{fig:P2_descriptionA} and \ref{fig:P2_descriptionB}, we propose allowing a commodity to be routed in $D_S$ via {\color{myred}$a$} if the subsequent timed arc is $((v_2, 1), (v_3, 3))$, and {\color{myblue}$a'$} if the subsequent timed arc is $((v_2, 0), (v_4, 2))$. Note that this approach has potential issues:
\begin{enumerate}[noitemsep]
\item[(I1)] Allowing different sets of departure times for arcs in $\delta^+(v)$ introduces additional timed copies of arcs in $\delta^-(v)$. For example, allowing different departure times for arcs $v_2v_3$ and $v_2v_4$ introduces two copies of arc $v_1v_2$ departing $(v_1, 0)$.
\item[(I2)] The savings from consolidating flows is no longer captured with the existing formulation. Two trajectories departing arc $uv$ at the same time may not travel on the same timed arc in $D_S$.
\end{enumerate}

The former issue is the main hurdle, since the objective in moving to a discretization for each arc was to decrease the number of variables in the lower bound formulations. However, recall that in a time-indexed formulation on $D_S$, there is a flow variable for the pair $(k, ((v,t), (w,t')))$ only when $vw \in D^k$. Returning to the running example in Figure \ref{fig:P2_description}, if for commodity $k$ we have $v_2 v_3 \in D^k$ and $v_2 v_4 \notin D^k$, then there must be a flow variable for the pair $(k, {\color{myred}a})$, but not for $(k, {\color{myblue}a'})$ since {\color{myblue}$a'$} is only permitted if the subsequent arc taken is $v_2v_4$. As a result, while the partial network includes additional timed arcs incoming to $v_2$ this \emph{does not} increase the number of variables for commodity $k$. 

This observation motivates the partitioning of the arc set presented in Section \ref{sec:arc-partition}, where for each $k \in \K$, all arcs in $\delta_{D^k}^+(v)$ are in the same part. We then use this arc partition to build an auxiliary flat network, denoted $G$, based on the original flat network $D$. In Section \ref{sec:G_T} we model SND-RR on the corresponding full network of $G$, denoted $G_T$, and modify the consolidation constraints in order to correct for (I2). In Theorem \ref{theorem:arc_based_equiv}, we prove that the resulting formulation is isomorphic to the original formulation \ref{IP:D_T}. In Section \ref{sec:arc-based}, we present a DDD algorithm that uses time-indexed formulations on the auxiliary graph $G$ that avoids issue (I1).

%%%%%%%%%%%%%%%%%%%%%%%%%%%%%%%%%%%%%%%%
%%%%%%%%%%% Auxiliary graph %%%%%%%%%%%%
%%%%%%%%%%%%%%%%%%%%%%%%%%%%%%%%%%%%%%%%
\subsection{Auxiliary graph}\label{sec:arc-partition}
Let $\I$ be an instance of SND-RR with base network $D = (N,A)$, time horizon $T$, and designated flat network $D^k \subseteq D$ for each commodity $k \in \K$. For each node $v \in N$, we partition the set of arcs in $\delta^+_D(v)$ into arc sets $\A_v = \{A_1, A_2, \ldots, A_{r}\}$ so that for each commodity $k \in \K$, $\delta^+_{D^k}(v) \subseteq A_i$ for some $A_i \in \A_v$. That is, the arcs departing $v$ that can be used by a fixed commodity in a feasible solution must be contained in a single set in the partition. Using this arc partition, in Section \ref{sec:arc-based} we define a DDD approach so that arcs in different sets in $\A_v$ can have different sets of available departure times in each iteration of DDD.

\begin{definition}
A partition $\A= \{\A_v\}_{v \in N}$ of the arcs in $D=(N,A)$ is \emph{valid with respect to the instance $\I$} if for each $v \in N$, 
\begin{enumerate}[noitemsep]
\item $\A_v = \{A_{1}, \ldots, A_{r}\}$ is a partition of $\delta^+_D(v)$, and 
\item for each commodity $k \in \K$, $\delta^+_{D^k}(v) \subseteq A_i$ for some $A_i \in \A_v$. 
\end{enumerate} 
\end{definition}
For each node $v \in N$, let $\V(v):= \{v^i: v \in N, i \in \{0, 1, \ldots, |\A_v|\}\}$ denote a set of nodes obtained by creating a copy of $v$ for each set $A_i \in \A_v$, along with a copy $v^0$ which we refer to as the \emph{terminal copy}. Given a valid partition of the arcs, $\A$, and flat network $D$, we construct an auxiliary graph where we create a copy of each node $v$ for each part in the partition of $\delta^+(v)$. The arc set is constructed so that there is a copy of each arc $vw$ departing the node copy representing the part containing $vw$, for each copy of $w$. This construction is stated more precisely in the following definition.

\begin{definition}
Given a flat network $D$ and a valid partition of the arc set $\A$, the \emph{auxiliary flat network $G(D, \A)$} has node set 
\[V :=  \bigcup_{v \in N}\V(v),\]
and arc set 
\[ E = \{v^i w^j: vw \in A_i \in \A,  w^j \in \V(w)\}. \]
\end{definition} 

When the base graph $D$ and arc partition $\A$ are clear from context, we simply write $G$ rather than $G(D, \A)$. In Figure \ref{fig:aux_graphB} we provide the auxiliary flat network when $D$ is the flat network in Figure \ref{fig:aux_graphA} and the arc partition is $\A = \{vw\}_{vw \in A}$ (each arc is in its own part). Observe that for each node $v$ in $D$, the corresponding copy $v^0$ in $G$ has no outbound arcs, and will only be used by commodities with destination at $v$. 
\begin{figure}[htbp]
    \centering
    \begin{subfigure}[b]{0.49\textwidth}
	\centering
        \includegraphics[width=.45\textwidth]{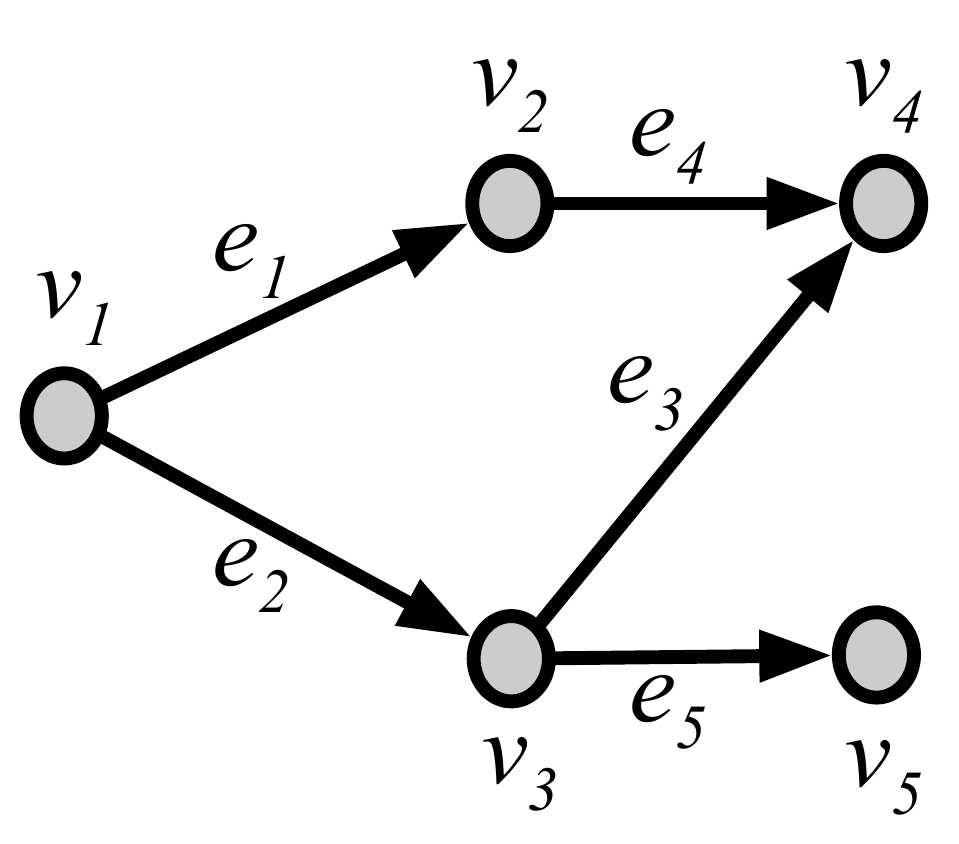}
        \caption{A flat network $D$.}
        \label{fig:aux_graphA}
    \end{subfigure}
    \begin{subfigure}[b]{0.49\textwidth}
	\centering
        \includegraphics[width=.6\textwidth]{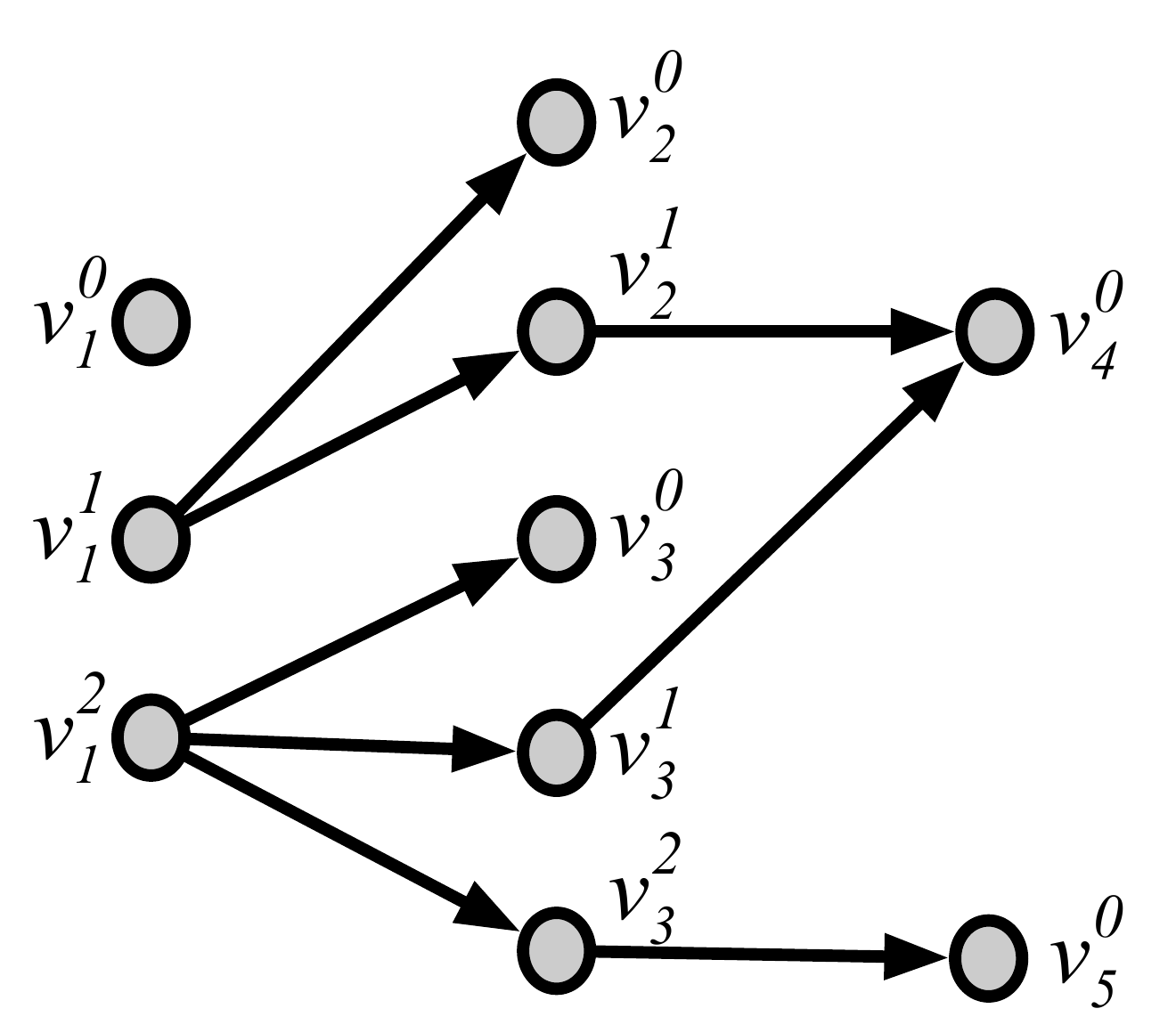}
        \caption{$G(D, \A)$ when $\A = \{vw\}_{vw \in A}$}
        \label{fig:aux_graphB}
    \end{subfigure}
    \caption{Construction of the auxiliary flat network}
    \label{fig:aux_D}
\end{figure}
\FloatBarrier

 We assign the fixed and variable costs as well as the transit time for each arc copy in $E$ to be the same as its underlying arc in $A$.  The format of the auxiliary network allows us to store the set of available departure times at the node copy $v^i$ for each set of arcs in part $A_i \in \A_v$, rather than forcing all arcs departing a fixed vertex to have the same set of departure times in each iteration of DDD. 

%%%%%%%%%%%%%%%%%%%%%%%%%%%%%%%%%%%%%%%%
%%%%% associated flat network in G %%%%%
%%%%%%%%%%%%%%%%%%%%%%%%%%%%%%%%%%%%%%%%
\subsection{Modeling SND-RR on $G_T$}\label{sec:G_T}
Recall that in the problem of SND-RR, each commodity has a designated feasible subgraph $D^k$ of $D$ through which it must be routed. That is, the physical route of commodity $k$ must be a path from $o_k$ to $d_k$ in $D^k$. In order to create a time-indexed formulation where the flat network is the auxiliary graph $G$, we first define the designated subgraph $G^k$ in $G$ for each commodity $k \in \K$. We then assign the origin and destination for each commodity in $G$. 

\subsubsection*{Mapping $D^k$ to $G^k$} 
Let $\A = \{\A_v\}_{v \in N}$ be the valid arc partition, and for each $k \in \K$ let $D^k = (N^k, A^k)$ denote the designated subgraph for commodity $k$. For each $k \in \K$, we may assume that $\delta^+_{D^k}(d_k) = \emptyset$, and for all $v \in N^k \setminus d_k$, $\delta^+_{D^k}(v) \neq \emptyset$ as otherwise $D^k$ could be reduced without losing any optimal solutions. 
 
For each set $A_i \in \A$, let $\K(A_i)$ denote the set of commodities that could traverse an arc in $A_i$ in a feasible solution. That is, for each $v \in N$ and each set of arcs $A_i \in \A_v$, let 
\[\K(A_i):= \{k \in \K: \emptyset \neq \delta_{D^k}^+(v) \subseteq A_i\}.\] 
By definition of a valid partition, for each commodity $k \in \K$, at each node $v \in N^k \setminus \{d_k\}$ there is a single copy of node $v$ in $G$ that has outgoing arcs that could be used by commodity $k$. More precisely, there is a single set $A_i \in \A_v$ such that $k \in \K(A_i)$. For the destination $d_k$, we assign the corresponding node to be $d_k^0$ in $G$.

For each commodity $k \in \K$ we define the designated subgraph $G^k$ in $G$ as the graph with node set
\[ V^k := \{v^i: v \in N^k, A_i \in \A_v, k \in \K(A_i)\} \cup \{d_k^0\}, \]
and arc set 
\[ E^k := \{v^i w^j: vw \in A^k, v^i \in V^k, w^j \in V^k\}. \]
Observe that there is a single copy of the origin and destination nodes, $o_k$ and $d_k$, in $V^k$. Let $o_k'$ and $d_k^0$ denote these node copies respectively.

We provide an example of this construction in Figure 
\ref{fig:aux3_D}, where we continue with the flat network $D$ and auxiliary network $G(D, \A)$ from Figure \ref{fig:aux_D}. In Figure \ref{fig:aux_graph3A} we show the subgraph $D^k$ in purple dashed lines, for a commodity with origin $v_1$ and destination $v_4$. In 
Figure \ref{fig:aux_graph3B} the subgraph $G^k$ is presented in purple dashed lines, and we see the origin and destination in $G^k$ are $o'_k = v_1^2$ and $d^0_k = v_4^0$ respectively. 

\begin{figure}[htbp]
    \centering
    \begin{subfigure}[b]{0.49\textwidth}
	\centering
        \includegraphics[width=.45\textwidth]{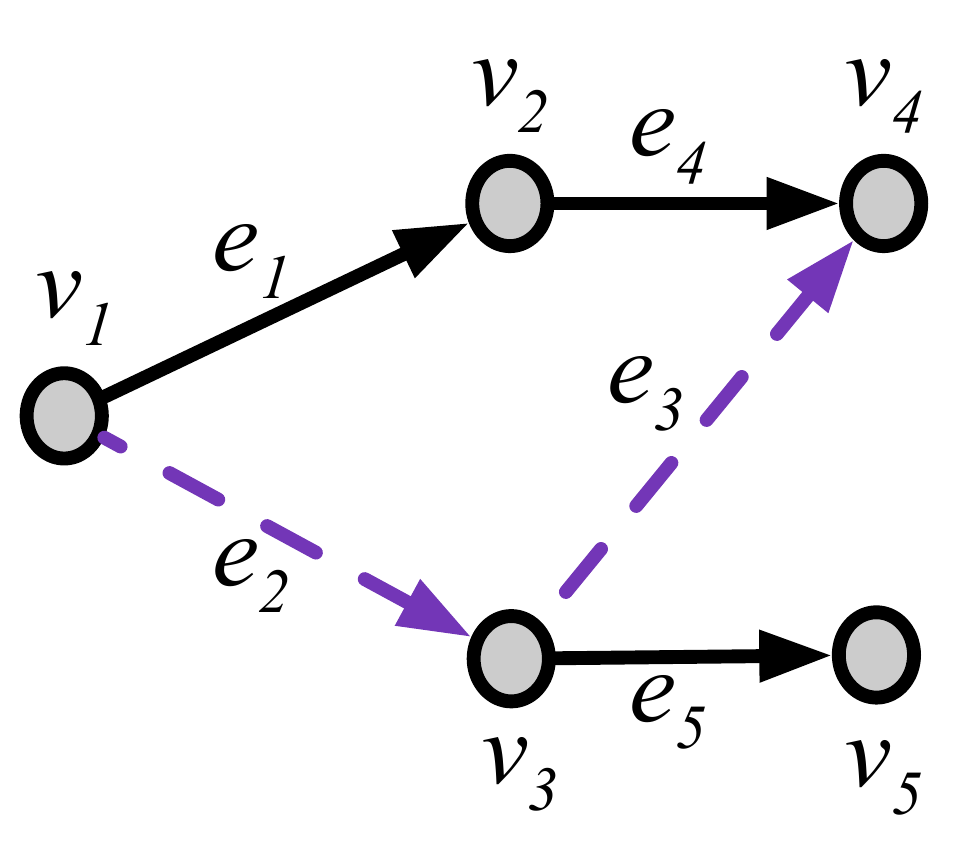}
        \caption{A flat network $D$ and feasible subgraph $D^k$.}
        \label{fig:aux_graph3A}
    \end{subfigure}
    \begin{subfigure}[b]{0.49\textwidth}
	\centering
        \includegraphics[width=.6\textwidth]{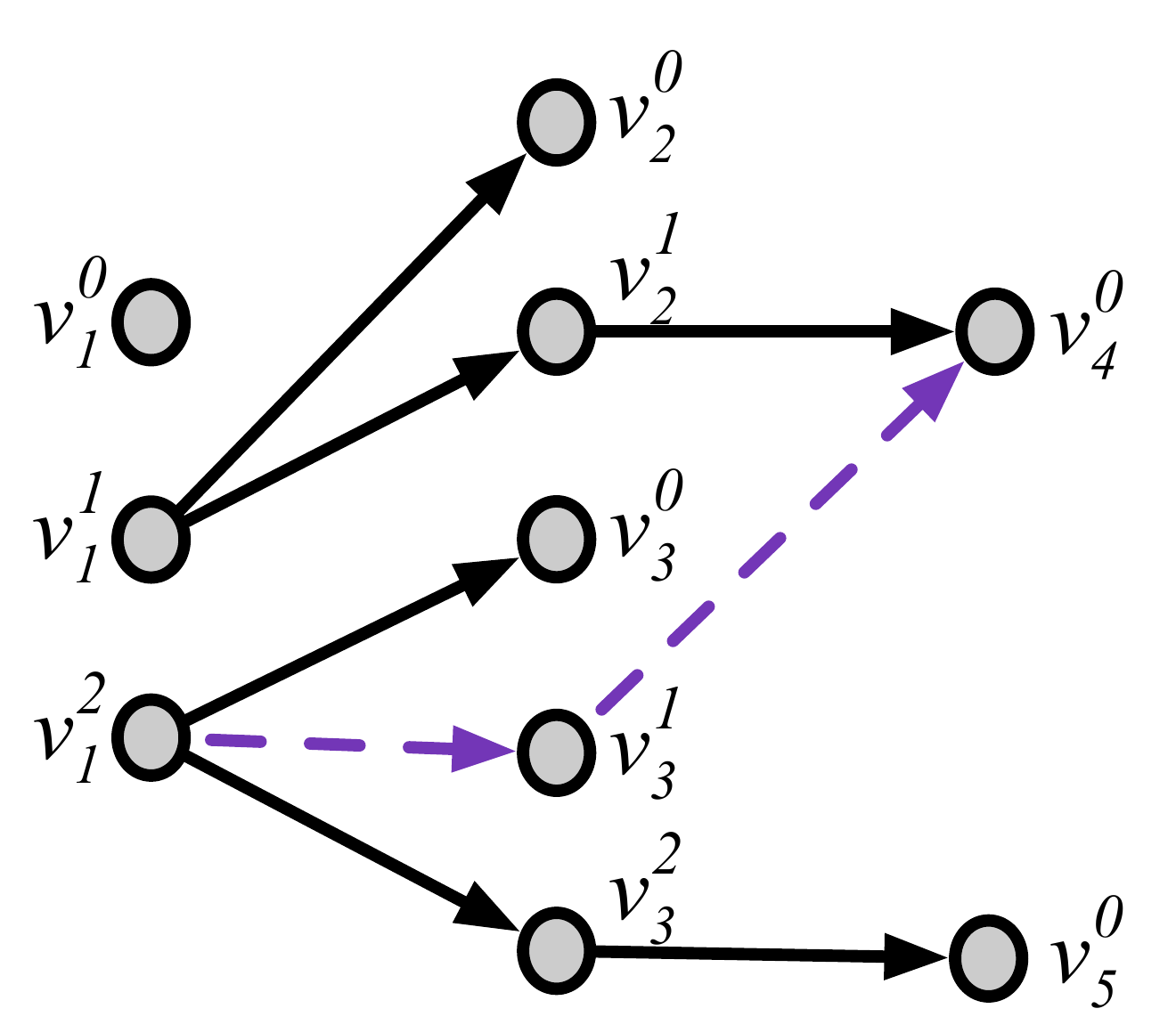}
        \caption{The corresponding subgraph $G^k$ in $G$.}
        \label{fig:aux_graph3B}
    \end{subfigure}
    \caption{Map from $D^k$ to $G^k$ when $D^k$ is a single path.}
    \label{fig:aux3_D}
\end{figure}
\FloatBarrier

\begin{lemma}\label{lemma:G_K_size}
For each commodity $k \in \K$, $G^k$ has the same number of nodes and arcs as $D^k$. 
\end{lemma}
\begin{proof}
Let $k \in \K$. First, we show that there is a one-to-one correspondence between the vertices in $D^k$ and the vertices in $G^k$. Let $v \in N^k$, and suppose $v \neq d_k$. Since $\A$ is a valid partition of the arcs and $\delta_{D^k}^+(v) \neq \emptyset$, $k \in \K(A_i)$ for a single set $A_i \in \A_v$. By definition of $V^k$ there is a single copy, $v^i$, of $v$ in $V^k$. Suppose instead that $v = d_k$. Then $\delta_{D^k}^+(v) = \emptyset$ and so $k \notin \K(A)$ for all $A \in \A_v$. As a result, the only copy of $v$ in $G$ is $d_k^0$. Thus, the number of nodes is the same in $D^k$ and $G^k$.

Since there is a single node in $V^k$ for each node in $N^k$, there is one copy $v^i w^j$ in $E^k$ of each arc $vw$ in $A^k$. Therefore the number of arcs in each graph is also the same.
\end{proof}

Furthermore, there is a natural bijection between the set of dipaths in the flat networks $D$ and $G(D, \A)$. Let $\P_D$ denote the set of dipaths in $D$, and let $\P_G$ denote the set of dipaths in $G$ which end at a terminal node. We define a map $\phi: \P_D \rightarrow \P_G$ so that for any dipath $P \in \P_D$ with node sequence $\{v_1, v_2, \ldots, v_{j+1}\}$ and arc sequence $\{a_1, a_2, \ldots, a_{j}\}$, then $\phi(P)$ is the unique dipath in $G$ with the node sequence $\{v^{i_1}_1, v^{i_2}_2, \ldots, v^{i_j}_j, v_{j+1}^0\}$, where for each $\ell \in [j]$, $a_l$ is in $A_{i_l} \in \A$. We think of $\phi$ as both as a map of the vertices and arcs of a trajectory. Let $\phi^{-1}$ denote the inverse of $\phi$. That is, $\phi^{-1}$ takes a path $P$ in $G$ and projects each arc (node) copy in $G$ down to its corresponding arc (node) in $D$.

\begin{fact}\label{fact:bijection_D}
$|\P_D| = |\P_G|$, and  $\phi$ is a bijection between $\P_D$ and $\P_G$. 
\end{fact}

Given a time horizon $T$, let $G_T = (V_T, E_T \cup F_T)$, denote the full network with respect to $G$ and $T$, where $E_T$ denotes the movement arcs and $F_T$ denotes the storage arcs. Let $G_T^k$ denote the full network corresponding to $G_k$. Fact \ref{fact:bijection_D} naturally extends to the following analogous statement for $D_T$ and $G_T$. Let $\Q_D$ denote the set of trajectories in $D_T$, and let $\Q_G$ denote the set of trajectories in $G_T$ that end at a timed terminal node. 

We define the map $\phi_T: \Q_D \rightarrow \Q_G$ as follows (again, we will think of $\phi_T$ as mapping both timed arcs and timed nodes). Let $Q \in \Q_D$ be a trajectory in $D_T$ and let $P$ be the corresponding flat path in $D$. Note that any trajectory in $D_T$ is is fully determined by its physical path in $D$ along with a departure time for each arc in the path. We define $\phi_T(Q)$ as the trajectory in $G_T$ defined by the path $\phi(P)$ along with the departure times of $Q$. Let $\phi^{-1}_T$ denote the inverse of $\phi_T$. That is, $\phi^{-1}_T$ takes a trajectory $Q$ in $G_T$ and projects each timed arc (timed node) copy in $G_T$ down to its corresponding timed arc (timed node) in $D_T$.

\begin{fact}\label{fact:bijection_G}
$|\Q_D| = |\Q_G|$, and $\phi_T$ is a bijection between $\Q_D$ and $\Q_G$. 
\end{fact}

These bijections allow us to argue that the following formulation for SND-RR is correct. 

%%%%%%%%%%%%%%%%%%%%%%%%%%%%%%%%%%%%%%%%
%%%%%%%%% Larger formulation %%%%%%%%%%%
%%%%%%%%%%%%%%%%%%%%%%%%%%%%%%%%%%%%%%%%

\subsubsection*{The formulation SND-RR($G_T$)}
In the definition of $G$, there are multiple copies of the same arc $vw \in A$. Suppose two commodities $k$ and $k'$ traverse the arc $vw$ at the same time, but then travel along arcs departing $w$ in different parts according to the partition $\A$. As a result, while it was possible for each commodity to traverse the same timed copy of $vw$ in $D_T$, it is \emph{no longer possible} in $G_T$. To capture the ability to consolidate flow in $G_T$, we define subsets of timed arcs in $G_T$ which have the same departure time and underlying arc in $D$ (as opposed to $G$). That is, for each timed arc $a = ((v,t),(w, t')) \in A_T$, we define the set 
\[ \E_T(a):= \{((v^i, t), (w^j, t')): vw \in A_i \in \A_v, w^j \in \V(w)\}.\] 
We now present a formulation for solving instances of SND-RR, denoted \ref{IP:G_T}, which is defined on the time-expanded network $G_T$. We will continue to use $a$ to denote timed arcs in $D_T$, and will use $e$ to denote timed arcs in $G_T$. Note that there are a few key differences between \ref{IP:D_T} and \ref{IP:G_T}. First, $G_T$ is a larger graph, even in the case where $\A$ is the trivial partition, $\{\delta^+(v)\}_{v \in N}$, in which case $G_T$ is roughly double the size of $D_T$ (there are two copies of each node -- the copy with departing arcs, and the terminal copy). We also restrict commodities to the designated subgraphs $G^k_T$ in $G_T$ rather than in $D_T$. Finally, we capture the consolidation onto different arc copies through modifying constraint (\ref{const:G_Tcapacity}) for all $a \in A_T$ by summing over each timed arc copy in $\E_T(a)$.
\begin{align}\tag{SND-RR$(G_T)$}\label{IP:G_T}
    \min~~ & \sum_{a \in A_T} f_a y_a + \sum_{k \in \K} \sum_{\substack{e \in E_T}}  c^k_e q_k x_{e}^k \\
    \vspace{.5cm}
    \mbox{s.t.} ~~
	& x^k(\delta_{G_T^k}^{+}(v,t)) - x^k(\delta_{G_T^k}^{-}(v,t)) = 	\begin{cases}
        1 ~(v,t) = (o'_k, r_k) \\
        -1 ~(v,t) = (d_k^0, l_k) \\
        0 ~\mbox{otherwise} \\
    \end{cases} \quad \forall k \in \mathcal{K}, (v,t) \in V^k_T \label{const1:flow_routing}\\
     & \sum_{k \in \mathcal{K}} \sum_{e \in \E_T(a)} q_k x_{e}^k \leq u_a y_a \quad \forall a \in A_T \label{const:G_Tcapacity}\\
    & x_e^k \in \{0,1\} \quad \forall k \in \mathcal{K}, e \in E_T^k \cup F_T^k \label{var:G_Tflow} \\
	& y_a \in \mathbb{N}_{\geq 0} \quad \forall k \in \mathcal{K}, a \in A_T \label{var2:G_Tflow}
\end{align}

However, despite working with a larger graph $G_T$, the number of variables and constraints are the same in \ref{IP:D_T} and \ref{IP:G_T}. 

\begin{theorem}\label{theorem:arc_based_equiv}
\ref{IP:D_T} and \ref{IP:G_T} have the same number of variables and constraints, and are equivalent. 
\end{theorem}

\begin{proof}
By Lemma \ref{lemma:G_K_size}, for each $k \in \K$, $D^k$ and $G^k$ have the same number of nodes and arcs. This implies that $D_T^k$ and $G_T^k$ also have the same number of timed nodes and timed arcs since the transit time of an arc in $G$ is the same as the corresponding arc in $D$. Thus, the number of variables and constraints in \ref{IP:D_T} and \ref{IP:G_T} is the same. 

We now prove that there is a cost-preserving bijection between feasible solutions in \ref{IP:D_T} and feasible solutions in \ref{IP:G_T}. Let $(\bar{x}, \bar{y})$ be a feasible solution in \ref{IP:D_T} with corresponding trajectories $\Q = \{Q_k\}_{k \in \K}$ in $D_T$ and paths $\P = \{P_k\}_{k \in \K}$ in $D$. For each $k \in \K$, consider the set of trajectories $\phi_T(\Q) = \{\phi_T(Q_k)\}_{k \in \K}$. Since the physical path of $Q_k$ is $\phi(P_k)$, and $P_k \subseteq D^k$, it follows from the definition of $G^k$ that $\phi(P_k) \subseteq G^k$. Additionally, $\phi(P_k)$ has origin $o'_k$ and destination $d^0_k$. Finally, since the departure times are the same for arcs in $P_k$ and their corresponding arc copies in $\phi(P_k)$, it follows that $\phi(Q_K)$ is feasible for commodity $k$ in \ref{IP:G_T}. 

It remains to argue that the cost of $\phi_T(\Q)$ in \ref{IP:G_T} is the same as the cost of $\Q$ in \ref{IP:D_T}. Since arc copies in $G$ are assigned the same costs as the underlying arcs in $D$ and $\phi(P)$ simply assigns each arc in $P$ to one of its copies in $G$, the variable cost for each trajectory is the same for $\Q$ and $\phi_T(\Q)$. We now consider the fixed cost for timed arc $a \in A_T$. Since $\Q$ and $\phi_T(\Q)$ have the same set of departure times and $\phi$ maps each arc in the flat path to one of its copies in $G$, it follows that the same set of commodities $k$ appear in the capacity constraint for timed arc $a$ in \ref{IP:D_T} and \ref{IP:G_T}. Therefore the fixed cost of $\phi_T(\Q)$ in \ref{IP:G_T} is also the same as the fixed cost of $\Q$ in \ref{IP:D_T}. The reverse direction is analogous. 
\end{proof}

Thus, in order to solve \ref{IP:D_T}, we can instead solve \ref{IP:G_T}. As previously mentioned, moving to \ref{IP:G_T} has no advantage over \ref{IP:D_T} if we are solving the formulations directly with a MIP solver since they are isomorphic. However, we will show that the DDD paradigm can be improved when applied to the base graph $G$ rather than $D$. Note that even if $D^k = D$ for some (or all) $k \in \K$, the auxiliary network still splits each node in $D$ into two nodes (one for departing arcs, and the other denoting the terminal copy).

\section{An arc-based DDD approach}\label{sec:arc-based}

%%%%%%%%%%%%%%%%%%%%%%%%%%%%%%%%%%%%%%%%
%%%%%% DDD on auxiliary network %%%%%%%%
%%%%%%%%%%%%%%%%%%%%%%%%%%%%%%%%%%%%%%%%

We now define a new DDD algorithm based on a corresponding lower bound formulation for \ref{IP:G_T}. We describe the lower bound formulation, upper bound, and refinement processes required for a complete DDD algorithm. We will refer to the DDD algorithm based on the auxiliary network as an \emph{arc-based} DDD algorithm, and the original approach as a \emph{node-based} DDD algorithm. The arc-based DDD approach is not simply the standard DDD algorithm applied to this new auxiliary network. Instead, the approach must be modified to ensure costs continue to capture all possible consolidation opportunities. 

The advantage of using the auxiliary network is that in the corresponding DDD algorithm, each set of arcs in the same part of the arc partition has its own set of departure times. As a result, in each iteration fewer variables and constraints need to be added in order to improve the current network while maintaining a guaranteed lower bound. Specifically, when lengthening a short arc $((v,t),(w,t'))$, we can now add the departure time $t + \tau_{vw}$ to a subset of the arcs departing $w$, rather than the entire set. In Section 6, we highlight two particular applications where the arc-based approach has the potential to generate smaller iterations than the node-based approach.

In this section we assume $\A$ is a valid partition of the arc set $A$, and $G = G(D, \A)$ is the corresponding auxiliary network. Furthermore, each commodity $k \in \K$ has a designated network $D^k \subseteq D$, and designated network $G^k \subseteq G$. Additionally, the origin and destination of commodity $k$ in $G$ are denoted $o_k'$ and $d_k^0$ respectively. 

\subsection{Lower bound model}
Just as in the case of applying DDD to the base graph $D=(N,A)$, we require restrictions on the structure of a partial network of $G = (V,E)$, denoted $G_S = (V_S, E_S \cup F_S)$, and also require an accompanying formulation. We first begin with restrictions on $G_S$. 

\subsubsection*{Properties to guarantee a lower bound}
We will prove that when a partial network $G_S = (V_S, E_S \cup F_S)$ satisfies the following two properties, the optimal value of the corresponding formulation \ref{IP:G_S} defined in Section \ref{sec:LB_G_S} gives a lower bound on the value of \ref{IP:G_T}. These properties are analogous to properties (P1) and (P2) introduced in Section \ref{sec:background}. 
\begin{itemize}[noitemsep]
            \item[(P1$'$)]\textbf{Timed nodes}:

\hspace{1cm} For all $k \in \mathcal{K}$, $(o_k', r_k) \in V_S$ and $(d_k^0, l_k) \in V_S$;

\hspace{1cm} For all $v \in N$, $(v^0, T) \in N_S$ and $(u,0) \in N_S$ for all $u \in \V(v)$;
            \item[(P2$'$)]\textbf{Arc copies}: 

\hspace{1cm} For all $vw \in E$, for all $(v,t) \in V_S$ with $t + \tau_{vw} \leq T$, we have $((v,t), (w, t')) \in E_S$ 

\hspace{1cm} where $t' = \max \{r: r \leq t + \tau_{vw}, (w, r) \in V_S\}$;
        \end{itemize}
We define $F_S$ to be the set of holdover arcs connecting $V_S$. Given $G_S = (V_S, E_S \cup F_S)$, for each $k \in \K$ we obtain $G_S^k = (V_S^k, E_S^k \cup F_S^k)$ where $V_S^k = \{(v,t) \in V_S: v \in V^k\}$, $E_S^k = \{((v,t), (w, t')) \in E_S: vw \in E^k\}$, and $F_S^k$ is the corresponding set of holdover arcs.

\subsubsection*{Notation}
Let $D_S(G_S) = (N_S, A_S\cup H_S)$ denote the partial network on $D$ with timed nodes 
\[N_S = \{(v,t): (u, t) \in V_S \mbox{ for some } u \in \V(v)\},\] 
and $A_S$ constructed according to $(P2)$. Similar to the definition of $\E_T$, for each timed arc $a = ((v,t),(w, t'))$ in $A_S$, let $\E_S(a)$ denote the set of timed arc copies of $vw$ with departure time $t$. That is, 
\[ \E_S(a):= \{((v^i, t), (w^j, t'')) \in E_S: w^j \in \V(w), vw \in A_i \in \A_v\}.\] 
We may have $e = ((v^i, t), (w^j, t''))$ where $t'' \neq t'$ due to rounding down transit times in $G_S$. Note that each timed arc $e \in E_S$ is in a set $\E_S(a)$ for exactly one timed arc $a \in A_S$. 

\subsubsection{Lower bound formulation}\label{sec:LB_G_S}
We now state the lower bound formulation defined on valid partial auxiliary networks $G_S$. The set $A_S$ in the formulation is the set of timed movement arcs in $D_S(G_S)$. In the following formulation, for $e = ((v,t), (w,t')) \in E_S$, $\tau_e$ denotes the transit time of arc $vw$ rather than $t' - t$.  

\begin{align}\tag{SND-RR$(G_S)$}\label{IP:G_S}
    \min~~ & \sum_{a \in A_S} f_a y_a + \sum_{k \in \K} \sum_{e \in E_S}  c^k_e q_k x_{e}^k \\
    \vspace{.5cm}
    \mbox{s.t.} ~~
& x^k(\delta_{G_S^k}^{+}(v,t)) - x^k(\delta_{G_S^k}^{-}(v,t)) = 	\begin{cases}
        1 ~(v,t) = (o'_k, r_k) \\
        -1 ~(v,t) = (d_k^0, l_k) \\
        0 ~\mbox{otherwise} \\
    \end{cases} \quad \forall k \in \mathcal{K}, (v,t) \in V^k_S \label{const1:flow_routing_G_S}\\
    & \sum_{k \in \mathcal{K}} \sum_{e \in \E_S(a)} q_k x_{e}^k \leq u_a y_a \quad \forall a \in A_S \label{const:G_Scapacity}\\
	& \sum_{e \in E_S^k} \tau_e x_e^k \leq l_k - r_k \quad \forall k \in \K. \label{const:feasible_G_S} \\
    & x_e^k \in \{0,1\} \quad \forall k \in \mathcal{K}, e \in E^k_S \cup F^k_S \label{varG1:flow} \\
	& y_a \in \mathbb{N}_{\geq 0} \quad \forall k \in \mathcal{K}, a \in A_S \label{varG2:flow}
\end{align}
Constraint (\ref{const1:flow_routing_G_S}) ensures each commodity $k$ is assigned a feasible trajectory in $G_S^k$. Constraint (\ref{const:G_Scapacity}) is similar to the capacity constraint (\ref{const:G_Tcapacity}) in \ref{IP:G_T}, with the additional relaxation that flow can consolidate onto a common truck if it \emph{departs} along some copy of the base arc in $D$ at the same time. Observe that when $G_S = G_T$, it follows that \ref{IP:G_S} is the same as \ref{IP:G_T}.

We now prove that properties $(P1')$ and $(P2')$ are sufficient to ensure that \ref{IP:G_S} is a relaxation of \ref{IP:G_T}. The proof of this result is similar to the proof of Theorem 2 \cite{Boland1}, when given the base graph $G$ instead of $D$. However, careful attention is taken to track the consolidation of flow, which is pointed out near the end of the proof. In the following proof, while we are working with nodes in $V_T$ and $V_S$, we drop the superscripts denoting the copy of the node in $D$ for ease of notation until the end of the proof. 

\begin{theorem}\label{theorem:G_S_LB}
When $G_S$ satisfies properties $(P1')$ and $(P2')$, the optimal value of (\ref{IP:G_S}) provides a lower bound on the optimal value of (\ref{IP:G_T}). 
\end{theorem}

\begin{proof}
Let $(\bar{x}, \bar{y})$ be a feasible solution to (\ref{IP:G_T}) with cost $C$, and let $\mathcal{Q} = \{Q_k\}_{k \in \K}$ denote the corresponding set of trajectories. Let $\mu: E_T \rightarrow E_S$ be the map defined so that for each timed arc $e \in E_T$, 
\begin{equation}\label{mu_eq}
    e = ((v,t), (w,t')) \quad \rightarrow  \quad \mu(e) = ((v,\hat{t}), (w, \hat{t}')),
\end{equation}  where $\hat{t} = \max\{s: s \leq t, (v, s) \in V_S\}$, and $\hat{t}' = \max\{s: s \leq \hat{t} + \tau_{vw}, (v, s) \in V_S\}$. Observe that $\hat{t}'$ is dependent on $\hat{t}$ rather than $t'$. Furthermore, $\mu$ is well-defined since $(P1')$ dictates that there is a timed node for each non-terminal node in $V$ with time 0, and $(P2')$ ensures that there is a (unique) copy of each arc in $E$ departing the selected timed node that underestimates the correct transit time. 

We obtain trajectories $\hat{\mathcal{Q}} = \{\hat{Q}_k\}_{k \in \K}$ by mapping each movement arc $e \in E_T$ to $\mu(e) \in E_S$, and forming trajectories by adding holdover arcs. We now describe this process more precisely, and prove that the resulting trajectories are well-defined and feasible in $G_S$. Let $Q_k$ be the trajectory induced by $\bar{x}$ for commodity $k$. Then the set of ordered movement arcs in $Q_k$ is $\{e_1, e_2, \ldots, e_{q_k}\}$, where $e_i = ((v_i, t_i^{out}), (v_{i+1}, t_{i+1}^{in}))$ for each $i \in [q_k - 1]$, and since $Q_k$ is feasible, the following statements are true:
\begin{enumerate}[noitemsep]
\item[(S1)] $v_1 = o'_k$, and $t_1^{out} \geq r_k$
\item[(S2)] $v_{q_k} = d^0_k$, and $t_{q_k}^{in} \leq l_k$
\item[(S3)] $t_{i+1}^{in} = t_i^{out} + \tau_{v_i v_{i+1}}$ for all $i \in [q_k - 1]$
\item[(S4)] $t_i^{in} \leq t_i^{out}$ for all $i \in \{2, 3, \ldots, q_k - 1\}$
\end{enumerate}
We apply the map $\mu$ to each movement arc and obtain the ordered set of movement arcs $\{\mu(e_1), \mu(e_2), \ldots, \mu(e_{q_k})\}$ in $E_S$, where $\mu(e_i) = ((v_i, \hat{t}_i^{out}), (v_{i+1}, \hat{t}_{i+1}^{in}))$ for each $i \in [q_k - 1]$. 

\textbf{Claim 1}: $\hat{t}_i^{in} \leq \hat{t}_i^{out}$ for all $i \in \{2, 3, \ldots, q_k - 1\}$. 

Fix $i \in \{2, 3, \ldots, q_k - 1\}$. By definition, $\hat{t}_{i}^{in} = \max\{t: t \leq \hat{t}_{i-1}^{out} + \tau_{v_{i-1} v_{i}}, (v_i,t) \in V_S\}$ where $\hat{t}_{i-1}^{out} = \max\{t: t \leq t_{i-1}^{out}, (v_{i-1},t) \in V_S\}$. Thus, 
\[ \hat{t}_{i}^{in} \leq  \max\{t: t \leq t_{i-1}^{out} + \tau_{v_{i-1} v_{i}}, (v_i, t) \in V_S\} \]

Similarly, $\hat{t}_{i}^{out} = \max\{t: t \leq t_{i}^{out}, (v_i, t) \in V_S\}$. Since $t_{i}^{in} = t_{i-1}^{out} + \tau_{v_{i-1} v_{i}}$, and $t_{i}^{in} \leq t_{i}^{out}$, it follows that 
\[ \hat{t}_{i}^{out} \geq \max\{t: t \leq t_{i-1}^{out} + \tau_{v_{i-1} v_{i}}, (v_i, t) \in V_S\},\] 
which gives the desired result that $\hat{t}_i^{in} \leq \hat{t}_i^{out}$.

We add holdover arcs to form trajectories for each $k \in \K$, and this process is well-defined by Claim 1. Observe that each $\hat{Q}_k$ is contained in $G^k_S$, and by property $(P1')$ along with (S1) and (S2), we see that $\hat{t}_1^{out} \geq r_k$ and $\hat{t}_{q_k} \leq l_k$. Therefore $\hat{Q}_k$ is feasible in $G_S$ for commodity $k$. 

It remains to show that the cost of $\hat{\Q}$ in $D_S$ is at most the cost of $\Q$. Since the  underlying paths in the flat network for $\hat{\mathcal{Q}}$ and $\mathcal{Q}$ are the same for each commodity, the variable cost of $\hat{\Q}$ and $\Q$ is the same. Additionally, all trajectories in $\hat{\Q}$ have flat paths with feasible total transit time (at most $\ell_k - r_k$ for each $k \in \K$). 

We now prove that the fixed cost of $\hat{\Q}$ is at most the fixed cost of $\Q$. Observe that and any pair of commodities $k, k'$ traversing the same timed arc $e \in E_T$ now traverse the same timed arc $\mu(e) \in E_S$. It remains to prove that if two distinct timed arcs $e$ and $e'$ in $E_T$ are in $\E_T(a)$ for some timed arc $a \in A_T$, then $\mu(e)$ and $\mu(e')$ are in $\E_S(a')$ for some $a' \in A_S$.

Suppose $e$ and $e'$ are in $\E_T(a)$ for some $a = ((v,t), (w, t')) \in A_T$, and $vw \in A_i \in \A_v$. Then $e = ((v^i, t), (w^{j}, t'))$ and $e' = ((v^i, t), (w^{p}, t'))$, where $w^j, w^p \in \V(w)$. It follows that $\mu(e) = ((v^i, \hat{t}), (w^{p}, \hat{t}'))$ and $\mu(e') = ((v^i, \hat{t}), (w^{p}, \hat{t}''))$. Therefore, $\mu(e)$ and $\mu(e')$ are in the set $\E_S(a')$, where $a' = ((v,\hat{t}), (w, \hat{t} + \tau_{vw}))$. Thus, the fixed cost of $\hat{\mathcal{Q}}$ is at most the fixed cost of $\Q$. 
\end{proof}

Observe that when $G_S = G_T$, \ref{IP:G_S} is equal to \ref{IP:G_T}. To fully define the DDD approach, it remains to outline the upper bound/termination procedure and the refinement procedure.

%%%%%%%%%%%%%%%%%%%%%%%%%%%%%%%%%%%%%%%%
%%%%%%% UB and refinement for G_S %%%%%%
%%%%%%%%%%%%%%%%%%%%%%%%%%%%%%%%%%%%%%%%

\subsection{Upper bound and refinement}\label{sec:arc-based_UB}

Let $(\hat{x}, \hat{y})$ be an optimal solution to \ref{IP:G_S} with corresponding trajectories $\Q$ in $G_S$. In each iteration, we provide a feasible solution to \ref{IP:D_T} as well as a set of timed nodes to add to $G_S$ if the feasible solution is not optimal. As in the case of the original node-based DDD approach, the generation of the feasible solution in each iteration is found by solving a continuous formulation whose input is defined by the solution $(\hat{x}, \hat{y})$ to the lower bound model. Both the arc-based DDD and node-based DDD approaches solve the same continuous formulation, \ref{LP:UB}, and the methods only diverge in the definition of the inputs to the formulation. We want to define the input $(\P, \J, \K^F, \T^F)$ so that when \ref{LP:UB} has an optimal solution $(\bar{t}, \bar{\delta})$ with value 0, then $(\bar{t}, \P)$ defines an optimal solution to \ref{IP:D_T}. 

The first difference from the node-based approach is that the physical paths for $\Q$ are in $G$ rather $D$. In order to apply the continuous formulation from the node-based approach, the paths in $G$ are projected down to the original flat network $D$ via the map $\phi$. The second difference is in the definition of the sets in $\J = \{J_{vw}\}_{{vw} \in A}$, which capture the savings in fixed costs due to consolidation. In the original node-based DDD approach, 
the set of commodities that contribute to the same (fixed charge) capacity constraint of timed arc $a \in A_S$ in \ref{IP:D_S} are all commodities that traverse $a$. In \ref{IP:G_S}, commodities contribute to the same capacity constraint for timed arc $a = ((v,t), (w,t')) \in D_S(G_S)$ if they depart along any \emph{copy} of arc $vw$ at the same time. The construction of the input $(\P, \J, \K^F, \T^F)$ in each iteration is presented in Algorithm \ref{alg:UB_input_G_S}.

\begin{algorithm}
\DontPrintSemicolon
\KwIn{flat network $D=(N, A)$, partial (auxiliary) network $G_S = (V_S, E_S \cup F_S)$, corresponding partial network $D_S(G_S) = (N_S, A_S \cup H_S)$, and optimal solution $(\hat{x}, \hat{y})$ to \ref{IP:G_S}}
Let $\Q = \{Q_k\}_{k \in \mathcal{K}}$ denote the set of trajectories in $G_S$ given by $\hat{x}$\\
$\P \leftarrow \{P_k\}_{k \in \K}$, where $P_k = \phi(P'_k) \in D$ and $P'_k$ is the underlying path of $Q_k$ in $G$ \\
\For{$vw \in A$}{
	$J_{vw} := \{(k_1, k_2): \exists a = ((v,t),(w, t')) \in A_S, \exists f_1, f_2 \in \E_S(a)$ such that $\hat{x}^{k_1}_{f_1} = \hat{x}^{k_2}_{f_2}\}.$}
	$\J \leftarrow \{J_{vw}\}_{vw \in A}$\\
	$\K^F:= \{k \in \K: \forall f = ((v,t), (w,t')) \in A_S$ such that $\hat{x}_f^k > 0$, $t' = t + \tau_{vw}\}$\\
\For{$k \in \K^F$}{ 
$\T_k = \{\bar{t}_{v_p}^k\}_{p \in [|P_k|-1]}$, where $ \bar{t}_{v_p}^k := \{t: \exists u \in \V(v_p), \hat{x}^k(\delta_{E_S}^+(u, t)) > 0\}$. ie, the time commodity $k$ leaves a copy of $v_p$ in $\hat{x}$.
} 
$\T^F \leftarrow \{\T_k\}_{k \in \K^F}$\\
\Return $(\P, \J, \K^F, \T^F)$  
	\caption{$\mathtt{UB\mbox{-}input}(D, G_S, D_S(G_S), (\hat{x}, \hat{y})$)}
	\label{alg:UB_input_G_S}
\end{algorithm}	
\FloatBarrier

As in the case of the node-based approach, there is an optimal solution, $(\bar{t}, \bar{\delta})$, to \ref{LP:UB} in each iteration and so a feasible solution $(\bar{x}, \bar{y})$ to \ref{IP:D_T} is obtained in each iteration from $\bar{t}$ and $\P$. When $\bar{\delta} = \mathbf{0}$, we see that $(\bar{x}, \bar{y})$ is an optimal solution for \ref{IP:D_T}, since it has the same variable and fixed cost as $(\hat{x}, \hat{y})$. When the optimal value to \ref{LP:UB} is not equal to 0, we need to refine the partial network $G_S$. The refinement process is again analogous to the node-based approach. When the optimal $(\bar{t}, \bar{\delta})$ solution has non-zero value, $\bar{\delta}$ defines a set of commodities, $\C$, with infeasible trajectories. This process is presented in Algorithm \ref{alg:UB_G_S}, and we formally state this result in Theorem \ref{theorem:G_T_UB}. 

\begin{algorithm}
\DontPrintSemicolon
\KwIn{flat network $D=(N, A)$, partial (auxiliary) network $G_S$, partial network $D_S(G_S)$, and optimal solution $(\hat{x}, \hat{y})$ to \ref{IP:G_S}}
$(\P, \J, \K^F, \T^F) \leftarrow \mathtt{UB\mbox{-}input}(D, G_S, D_S(G_S), (\hat{x}, \hat{y}))$\\
Let $(\bar{t}, \bar{\delta})$ denote an optimal solution to \ref{LP:UB}\\
Let $\bar{Q}$ be the trajectories in $G_T$ defined by $\P, \bar{t}$ and let $(\bar{x}, \bar{y})$ be the corresponding solution to \ref{IP:G_T}.\\
	$\C = \{k_1 \in \K \setminus \K^F: \exists vw \in A, k_2 \in \K, \bar{\delta}_{vw}^{k_1 k_2} > 0\}$\\
\Return $(\bar{x}, \bar{y}), \C$ 
\caption{$\mathtt{UB}(D, G_S, D_S(G_S), (\hat{x}, \hat{y})$)}
\label{alg:UB_G_S}
\end{algorithm}	
\FloatBarrier

\begin{theorem}\label{theorem:G_T_UB}
Given an optimal solution $(\hat{x}, \hat{y})$  to \ref{IP:G_S} with corresponding trajectories $\Q$ in $G_S$, Algorithm \ref{alg:UB_G_S} either finds an optimal solution to \ref{IP:D_T}, or provides a feasible solution to \ref{IP:D_T} along with a nonempty set of infeasible trajectories $\C \subseteq \Q$ in $G_S$.
\end{theorem}

Similar to the node-based approach, in the refinement step, a timed node is added to $V_S$ to lengthen the earliest-departing short timed arc in each trajectory $Q_k$ where $k \in \C$. Since the timed nodes are added to $V_S$ rather than $N_S$, the increase in size of the subsequent lower bound formulation is typically less than the increase in size if the timed nodes were added to $N_S$.

%%%%%%%%%%%%%%%%%%%%%%%%%%%%%%%%%%%%%%%%
%%%%% overall arc-based DDD algo %%%%%%%
%%%%%%%%%%%%%%%%%%%%%%%%%%%%%%%%%%%%%%%%

\subsubsection{Pseudocode for the arc-based DDD approach}
The following algorithm states the overall arc-based DDD approach for solving SND-RR. 
\begin{algorithm}
	\DontPrintSemicolon
	\KwIn{Base network $D = (N,A)$, commodity set $\mathcal{K}$ with time horizon $T$}
	\textbf{Generate auxiliary network}: Let $G=(V, E)$ be the auxiliary network of $D$\\
\underline{\textbf{Initialization}}\\ 
$V_S \leftarrow V_0$, where $V_0$ is the minimal set of timed nodes satisfying $(P1')$\\
\While{not solved}{
\underline{\textbf{Lower bound}}\\ 
	Construct $G_S$ given updated $V_S$ using property $(P2')$.\\ 
	Solve \ref{IP:G_S} and obtain a solution $(\hat{x}, \hat{y})$ to \ref{IP:G_S}. \\
\underline{\textbf{Upper bound/termination}}\\
	Solve $\mathtt{UB}(D, G_S, D_S(G_S), (\hat{x}, \hat{y})$) for the set $\C$ and upper bound solution $(\bar{x}, \bar{y})$.\\
	\If{the cost of $(\bar{x}, \bar{y})$ is equal to the cost of $(\hat{x}, \hat{y})$ } {
            Stop. $(\bar{x}, \bar{y})$ is an optimal solution to SND($D_T$).}
\underline{\textbf{Refinement}}\\ For each $k \in \C$, lengthen the short timed arc $((v,t),(w, t')) \in Q_k$ with the earliest departure time $(t)$, by adding $(w, t+\tau_{vw})$ to $V_S$.\\
	}
	\caption{$\mathtt{SND\mbox{-}RR\mbox{-}arc\mbox{-}disc}$($D, \mathcal{K}$)}
	\label{alg:ddd_outline_aux}
\end{algorithm}
\FloatBarrier

\begin{theorem}
The arc-based DDD algorithm, Algorithm \ref{alg:ddd_outline_aux}, terminates with an optimal solution. 
\end{theorem}
\begin{proof}
The algorithm terminates when line 11 is executed. Since the cost of $(\hat{x}, \hat{y})$ gives a lower bound on the optimal value of \ref{IP:G_T} (Theorem \ref{theorem:G_S_LB}) and since $(\bar{x}, \bar{y})$ is feasible for \ref{IP:G_T} (Theorem \ref{theorem:G_T_UB}), if line 11 is executed then $(\bar{x}, \bar{y})$ is an optimal solution.

It remains to bound the number of iterations until line 11 is reached. If in an iteration line 11 is not reached, then the set $\C$ is nonempty (Theorem \ref{theorem:G_T_UB}) and consists of commodities with trajectories in $G_S$ with at least one short timed arc. As a result, in each iteration at least one timed (auxiliary) node is added to $V_S$. Thus, the algorithm terminates with an optimal solution in at most $|V|\cdot T$ iterations, where $|V| \leq 2 |A|$. 
\end{proof}

\section{Applications}\label{sec:applications}

%%%%%%%%%%%%%%%%%%%%%%%%%%%%%%%%%%%%%%%%
%%%%%%%%%% Designated paths %%%%%%%%%%%%
%%%%%%%%%%%%%%%%%%%%%%%%%%%%%%%%%%%%%%%%
\subsection{SND with designated paths}
First, we consider the problem of SND-RR where the designated flat network $D^k$ for commodity $k \in \K$ is a single path $P_k$ in $D$. This structure allows each arc to have an independent set of departure times in each iteration, as stated in the following Theorem. 

\begin{theorem}
For all instances $\I$ of SND-RR where $D^k$ is a single path for each commodity $k \in \K$, the arc partition $\A$ where each set $A \in \A$ is a single arc is valid for $\I$. 
\end{theorem}
\begin{proof}
Since each commodity has a designated path, for each node $v \in N$ and each $k \in \K$, $|\delta_{D^k}^+(v) \cap \delta_D^+(v)| \leq 1$. Thus, \emph{any} partition of the arc set $A$ is valid. 
\end{proof}

The partition of the arc set into singletons is the ideal scenario, since it allows each arc to have its own set of departure times in each iteration of arc-based DDD. The advantage of this fact can be seen already in the initial partial network. Consider the graph structure in Figure \ref{fig:star}, considered previously in Section \ref{sec:auxiliary}. For each $i \in [m]$, suppose there is a commodity $k$ with origin $v$, destination $v_i$, and release time $i$. We previously showed that the initial partial network for the node-based DDD approach would have a copy of node $v$ at all times in $[m]$. As a result, $D_0$ would have a copy of each arc for each departure time in $[m]$. Thus, we would have $A_0 \approx A_T$, and so the size of \ref{IP:D_S} is approximately the same as \ref{IP:D_T}. However, in the arc-based approach, $G_0$ only requires a single copy of arc $v v_i$ departing $v$ at time $i$ for each $i \in [m]$. Thus, the resulting formulation \ref{IP:G_S} is a factor $m$ smaller than \ref{IP:D_S}.

\begin{figure}[!htb]
\centering
\vspace{0cm}
\includegraphics[width=.22\textwidth]{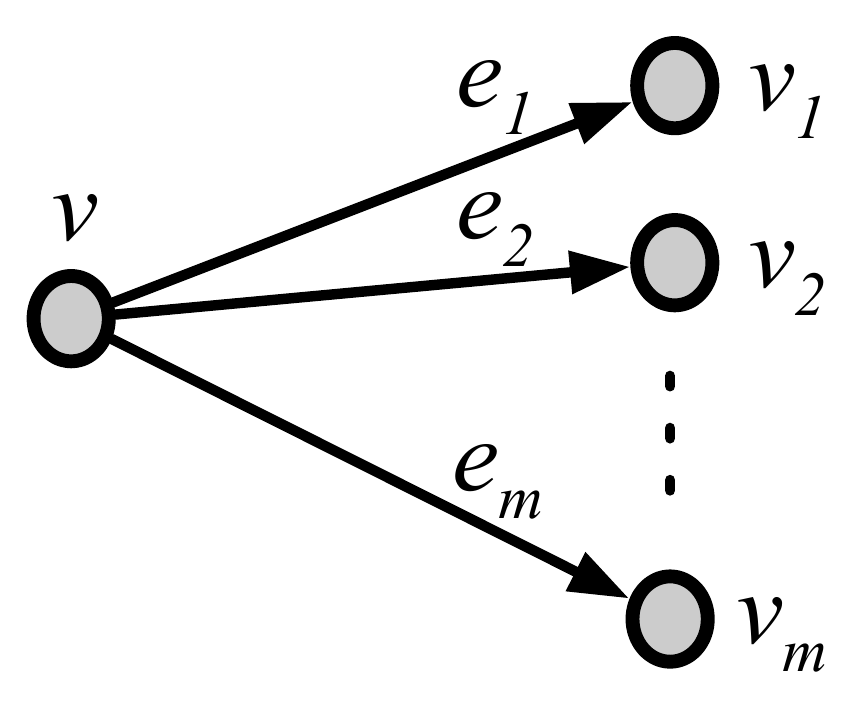}
\caption{}
\label{fig:star}
\end{figure}
\FloatBarrier

In Section \ref{sec:computational_results}, we compare the performance of the arc-based and node-based DDD algorithms on a  class SND-RR instances where each commodity has a designated path.

%%%%%%%%%%%%%%%%%%%%%%%%%%%%%%%%%%%%%%%%
%%%%%%%%%%%% Hub-and-spoke %%%%%%%%%%%%%
%%%%%%%%%%%%%%%%%%%%%%%%%%%%%%%%%%%%%%%%
\subsection{Hub-and-spoke networks}
We now consider hub-and-spoke, a popular region-based construction used in fulfillment and airline networks including FedEx and UPS \cite{bowen}. In a hub-and-spoke network, locations are divided into regions where each region is represented by a hub. Packages that have origins and destinations in two separate regions must travel between regions via the hubs. In Figure \ref{fig:Hub-and-spoke}, hubs are indicated as gray squares and only arcs for inter-region shipments are shown. 
\begin{figure}[h!]
\begin{center}
\includegraphics[scale=.4]{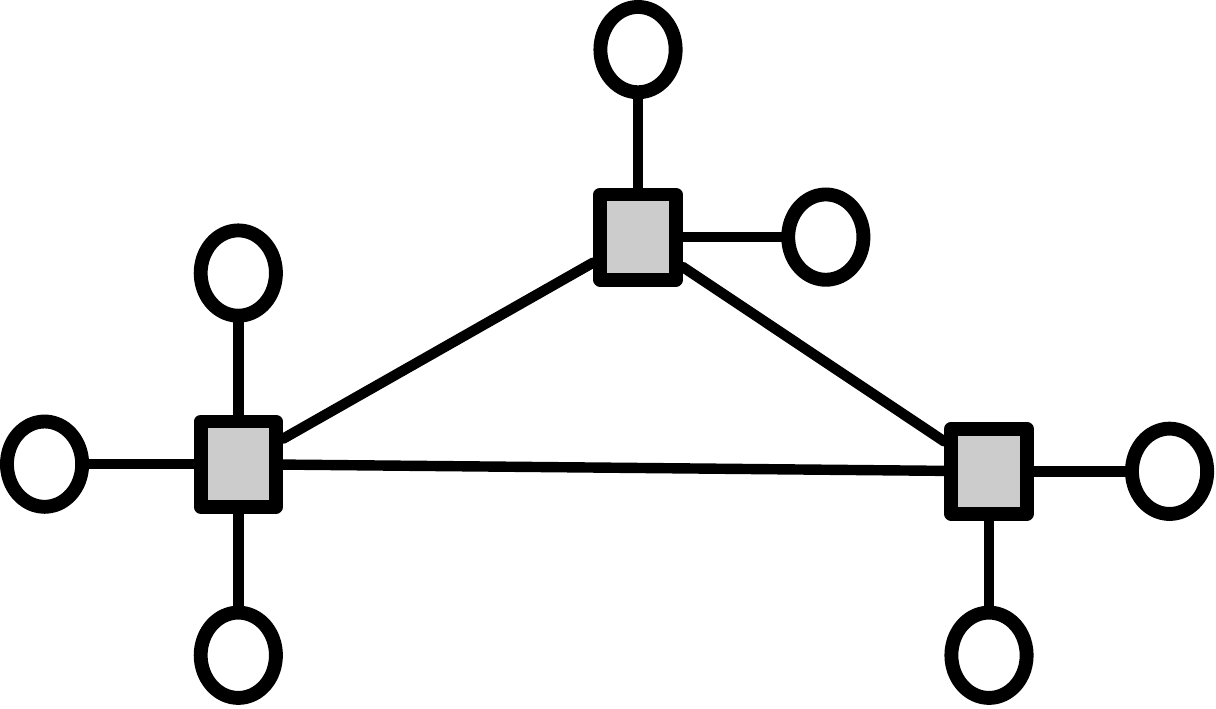}
\caption{Hub-and-spoke network}
\label{fig:Hub-and-spoke}
\end{center}
\end{figure}
\FloatBarrier

In fulfillment networks, it is often the case that packages have cut-off times in order to reach their destination warehouse for the next day. As a result, many intra-region arcs leaving a hub will only require a subset of the departure times. Namely, those departure times closely preceding the cut-off time. This points to the utility of allowing different departure times for arcs leaving hub nodes. To demonstrate that the arc-based DDD approach has advantages over the traditional node-based DDD approach, we generate a valid partition $\A$ of the arc set. First, we require additional notation. 

In a hub-and-spoke network, the node set $N$ is partitioned into a set of regions, $\R = \{R_1, R_2, \ldots, R_{\ell}\}$, with corresponding hubs $\H = \{h_1, h_2, \ldots, h_{\ell}\}$ where hub $h_i$ is in region $R_i$ for each $i \in [\ell]$. 
We refer to arcs as \emph{regional} if they connect two vertices within the same region, and \emph{national} if they connect two nodes in different regions (which must be hub nodes). Let $RA$ and $NA$ denote the set of regional and national arcs in $A$ respectively. Let $\R(v)$ denote the region of the node $v$, and let $\H(v)$ denote the hub node of the region containing $v$. 

For each commodity $k \in \K$, there is a natural restriction of the flat network when variable and fixed costs are nonnegative. Most notably, commodity $k$ will only potentially use regional arcs in $\R(o_k)$ and $\R(d_k)$. Furthermore, commodity $k$ will not use any national arc departing $\H(d_k)$ or entering $\H(o_k)$, or any regional arc departing $\H(o_k)$. Let $D^k \subseteq D$ consist of the regional arcs in $\R(o_k)$ and $\R(d_k)$, as well as all national arcs except those in $\delta_D^+(\H(d_k))$ and $\delta_D^-(\H(o_k))$. As discussed, restricting the route of commodity $k$ to $D^k$ does not increase the cost of an optimal solution. 

This definition of $D^k$ for each $k \in \K$ gives a natural partitioning of the arcs in $\delta_D^+(h)$ for each hub $h \in \H$.

\begin{theorem}
Let $\I = (D = (N,A), \K)$ be an instance of SND-RR where $D$ has a hub-and-spoke network structure with hubs $\H$, national arcs $NA$, and regional arcs $RA$. Let $\A = \{\A_v\}_{v \in N}$ be the arc partition where 
\begin{enumerate}[noitemsep]
\item $\A_h = \{A^h_{1}, A^h_{2}\}$ where $A^h_{1} = \delta_D^+(h) \cap RA$ and $A^h_{2} = \delta_D^+(h) \cap NA$ for all $h \in \H$
\item $\A_v = \{\delta^+_D(v)\}$ for all $v \in N \setminus \H$,
\end{enumerate} 
Then $\A$ is a valid partition of $A$ for $\I$. 
\end{theorem}

\begin{proof}
Let $k \in \K$, and consider a hub $h \in \H$. We may assume that $D^k$ is minimal in the sense that all arcs that cannot be used by commodity $k$ in an optimal solution have been removed. First suppose $\H(o_k) = \H(d_k) = h$. In this case, commodity $k$ will not use any national arcs departing $h$, and so $\delta_{D^k}^+(h) \subseteq A^h_1$ and $\delta_{D^k}^+(h') = \emptyset$ for all $h' \in \H \setminus \{h\}$. 

Suppose instead that $\H(o_k) \neq \H(d_k)$. For any hub $h \in \H \setminus \H(d_k)$, $\delta_{D^k}^+(h) \subseteq A^h_2$ since commodity $k$ will not use any regional arc departing $h$. If $h = \H(d_k)$, then similarly commodity $k$ will not use any national arc departing $h$, so $\delta_{D^k}^+(h) \subseteq A^h_1$. 
\end{proof}

We demonstrate the impact of this arc-discretization approach in Section \ref{sec:computational_results}.

\subsection{General SND instances}
In general, the best possible arc partition may be simply the partition $\A = \{\A_v\}_{v \in V}$, where $\A_v = \delta_D^+(v)$. However, the modified DDD approach still has potential to offer speed-up over the original implementation, since each node will still have two copies in the auxiliary network -- one representing the departing arcs, and the other representing the destination node. For settings where nodes serve as destinations for some commodities, and intermediate nodes for others, this still allows additional freedom in the selection of departure times, reducing the size of the partial networks in the DDD algorithm. 

We also note that given sets $D^k \subseteq D$ for each commodity $k \in \K$, the finest valid partition of the arc set can be computed efficiently. The following algorithm gives the pseudocode for this process. 

\begin{algorithm}
	\DontPrintSemicolon
	\KwIn{Base network $D = (N,A)$, and subgraph $D^k$ for each commodity $k \in \K$}
	$\A_v \leftarrow \{\{a\}: a \in \delta_D^+(v)\}$ for all $v \in N$\\
	\For{$k \in \K$}{
		\For{$v \in N$}{
			$\A_v' \leftarrow \{A_{i} \in \A_v: \delta_{D^k}^+(v) \cap A_{i} \neq \emptyset\}$\\
			Merge all sets in $\A'_v$ and update $\A_v$}
		}
	$\A \leftarrow \{\A_v\}_{v \in N}$ \\
	\caption{$\mathtt{arc\mbox{-}partition}(D, \{D^k\}_{k \in \K}$)}
	\label{alg:ddd_outline}
\end{algorithm}	
\FloatBarrier

Standard preprocessing approaches can be used to generate the subgraph $D^k$ for each $k \in \K$. For example, all arcs in $\delta_D^+(d_k)$ can be removed from $D$ for commodity $k$. Similarly, every arc $vw$ can be removed from $D$ for commodity $k$ if $D$ contains no $w, d_k$-dipath. In the same vein, arcs $vw$ can be removed from $D$ for commodity $k$ if $D$ contains no $o_k, v$-dipath.

\section{Computational Results}\label{sec:computational_results}

In this section we compare the performance of the novel arc-based DDD approach to the original node-based DDD approach. The two algorithms are applied to a variety of SND-RR instances in order to gain a better understanding of the factors that impact the comparative performance. The instances considered fall into three classes: SND-RR where each commodity has a designated flat path (Section \ref{sec:SND_paths}), SND where the flat network is a hub-and-spoke network (Section \ref{sec:SND_regional}), and SND on random flat networks where release times and deadlines are restricted to a set of critical times (Section \ref{sec:SND_CPT}). The latter two settings are instances of SND-RR where there are no restrictions on the physical route taken by a commodity.  

In Section \ref{sec:SND} we provide an overview of the construction of the SND instances used in \cite{Boland1}, which we use for our instances of SND with designated paths, and instances of SND with critical times. The instances on hub-and-spoke networks require a different construction of the flat network which we present in Section \ref{sec:SND_regional}. Each algorithm was coded in Python 3.6.9 with Gurobi 8.1.1 \cite{gurobi}  as the optimization solver. The running time limit was set to 10,800 seconds (three hours) using the deterministic option of the solver and the instances were solved on three cores to within 1\% of optimality. The instances were run in a 64 cores 2.6GHz Xeon Gold 6142 Processor with 256GB RAM, running a Linux operating system. The generated instances as well as the implementation of the arc-based and node-based DDD approaches can be found at \href{https://github.com/madisonvandyk/SND-RR}{\textsf{https://github.com/madisonvandyk/SND-RR}}.

%%%%%%%%%%%%%%%%%%%%%%%%%%%%%%%%%%%%%%%%
%%%%%%%%%%%% SND baseline %%%%%%%%%%%%%%
%%%%%%%%%%%%%%%%%%%%%%%%%%%%%%%%%%%%%%%%

\subsection{SND baseline instances}\label{sec:SND}
We generate the instances as constructed by Boland et al.~\cite{Boland1}. These temporal instances were modified from the flat instances of Crainic et al.~\cite{Crainic2, Crainic1} that are frequently used as benchmarking instances for static SND. In this construction, arcs ($A$) and origin-destination pairs forming commodities ($\K$) are chosen randomly. The capacities ($u$), demands ($q$), and fixed and variable costs ($f$ and $c$) are then chosen independently at random from uniform distributions. Afterwards, these values are scaled so that the specified \emph{cost ratio}, $F$, and \emph{capacity ratio}, $C$, are achieved where $Q$ is the total demand and 
\[ F := |\K| \sum_{ij \in A} \frac{f_{ij}}{Q \sum_{k \in \K} \sum_{ij \in A} c_{ij}^k} 
\quad \mbox{and} \quad 
C := \frac{|A| Q}{ \sum_{ij \in A} u_{ij}}. \]
The values of these ratios considered in \cite{Crainic2, Crainic1} are $F \in \{0.01, 0.05, 0.1\}$ and $C \in \{1, 2, 8\}$. The parameters chosen to generate flat instances in our experiments are presented in Sections \ref{sec:SND_paths}, \ref{sec:SND_regional}, and \ref{sec:SND_CPT}, and differ due to the varying difficulty of each family of instances. For example, when a physical path is designated for each commodity, this significantly decreases the number of variables, and so the number of nodes, arcs, and commodities are increased in order to obtain instances that are sufficiently large. For the precise construction of the flat instances see \cite{Crainic2}.

Boland et al.~\cite{Boland1} create timed instances by assigning transit times $\tau$ to arcs, and generating release times and deadlines for the commodities. The arc transit times are simply a scaling of the fixed costs. Given the transit times $\tau$, they compute the average transit time of the shortest commodity path, $\L$. The release times are then chosen according to a normal distribution with mean $\L$ and standard deviation $\sigma_r$. The deadline is set to be $r_k + \L + p_k$, where $p_k$ is the level of flexibility and is chosen from another normal distribution with mean $\mu_p$ and standard deviation $\frac16 \mu_p$. While various discretization levels, $\Delta$, are studied in \cite{Boland1}, we use only $\Delta = 1$ since our computational study compares runtime and iteration sizes rather than the degradation of the optimal value in coarser discretizations. Again, the parameters considered vary for each instance family, and so they are presented in the beginning of each subsection. 

%%%%%%%%%%%%%%%%%%%%%%%%%%%%%%%%%%%%%%%%
%%%%%%%%% SND designated paths %%%%%%%%%
%%%%%%%%%%%%%%%%%%%%%%%%%%%%%%%%%%%%%%%%
\subsection{SND with designated paths}\label{sec:SND_paths}
For each set of parameters in Table \ref{tab:flat_DP} we generate a flat instance. Then, for each of the 32 flat instances, we create 6 SND-RR instances according to the parameters in Table \ref{tab:timed_DP}, for a total of 192 instances. To create instances of SND-RR where each designated subgraph $D^k$ is a single path $P_k$, for each commodity $k \in \K$ we assign the path $P_k$ to be the a shortest $o_k, d_k$-dipath by transit time. We note that while variable costs could be removed from the model when paths are fixed, higher variable costs relative to fixed costs allow the algorithms to terminate earlier, and so this comparison is still valuable.

\begin{table}[h!]
\parbox{.45\linewidth}{
\centering
\resizebox{.4\textwidth}{!}{\begin{tabular}{c|c|c|c|c}
\toprule
$|N|$ &  $|A|$ & $|\K|$ & $F$ & $C$ \\
\hline 
20 & $230, 300$ & $150, 200$ & 0.05, 0.1 & 1, 8 \\
25 & $360, 480$ & $250, 300$ & 0.05, 0.1 & 1, 8 \\
\bottomrule
\end{tabular}}%
\caption{Flat instance parameters.}
\label{tab:flat_DP}
}
\hfill
\parbox{.5\linewidth}{
\centering
\resizebox{.4\textwidth}{!}{\begin{tabular}{c||c|c}
\toprule
Normal distribution &  $\mu$ & $\sigma$ \\
\hline 
For generating $r_k$ & $\L$ & $\frac13\L, \frac16\L, \frac19\L $\\
For generating $p_k$ & $\frac14 \L, \frac38\L$ & $\frac16 \mu$\\
\bottomrule
\end{tabular}}%
\caption{Timed instance parameters.}
\label{tab:timed_DP}
}
\end{table}
\FloatBarrier

\subsubsection*{Computational results}
In Table \ref{tab:deciles_DP} we present the deciles for the runtime of each algorithm, and for deciles where the algorithms did not terminate within the time limit, we instead compare the optimality gap found after terminating the algorithms after three hours. For each decile we see that the arc-based DDD approach offers a significant improvement over the traditional node-based DDD approach. The majority of the instances (182 out of 192) were solved to within $1\%$ of optimality within three hours by both the arc-based and node-based DDD algorithms. Among these instances, the arc-based DDD approach completed in 42.6\% of the time of the node-based approach, representing an improvement of 57.4\% in runtime.

\begin{table}[h!]
\centering
\resizebox{.8\textwidth}{!}{
    \begin{tabular}{c|ccc|ccc}
\toprule
 & \multicolumn{3}{c}{ave. total runtime (s)} & \multicolumn{3}{c}{ave. optimality gap} \\
\midrule
Decile &  arc-based &  node-based  &  $\%$ improvement &  arc-based &  node-based  &  $\%$ improvement \\ \hline
0.1 &  42  &  138   &  69.2\% &  $<1\%$   &  $<1\%$  &  N/A   \\
0.2 &  75 &  239   &  68.5\% &  $<1\%$ &  $<1\%$ &  N/A \\
0.3 &  128 &  437   &  70.6\% &  $<1\%$ &  $<1\%$ &  N/A \\
0.4 &  173 &  629   &  72.4\% &  $<1\%$ &  $<1\%$ &  N/A \\
0.5 &  250 &  833   &  70.0\% &  $<1\%$ &  $<1\%$ &  N/A \\
0.6 &  451 &  1,188   &  62.0\% &  $<1\%$ &  $<1\%$ &  N/A \\
0.7 &  727 &  1,832   &  60.3\% &  $<1\%$ &  $<1\%$ &  N/A \\
0.8 &  1,252 &  3,328   & 62.4\% & $<1\%$ &  $<1\%$ &  N/A \\
0.9 &  3,142 &  7,247   &  56.6\% &  $<1\%$ &  $<1\%$ &  N/A \\
1 &  10,800 &  10,800   &  N/A &  5.87\% &  8.49\% &  30.8\% \\
\bottomrule
\end{tabular}
}
\vspace{-.25cm}
\caption{Runtime and optimality gap comparison for designated path instances.}
\label{tab:deciles_DP}%
\end{table}%
\FloatBarrier
The fraction of the instances solved over time is presented in Figure \ref{fig:cumulative_DP}, and we again observe a consistent improvement of the arc-based approach over the node-based approach. In Table \ref{tab:iteration_DP} we present the iteration statistics of each algorithm considering only the 182 instances which were solved within the time limit for both algorithms. The table reports the average number of iterations until termination, and the average number of variables and constraints in the \emph{final iteration}. The ratio for each field is the arc-based value to the node-based value. While the arc-based approach terminates after on average 18\% additional iterations of DDD, the runtime decreases since the formulation solved in each iteration in the arc-based approach is smaller; on average there are only 55\% the number of variables and 60\% the number of constraints in the final iteration compared to the node-based DDD approach.

\begin{figure}[h!]
\centering
\begin{minipage}{0.5\textwidth}
\centering
\includegraphics[width=\textwidth]{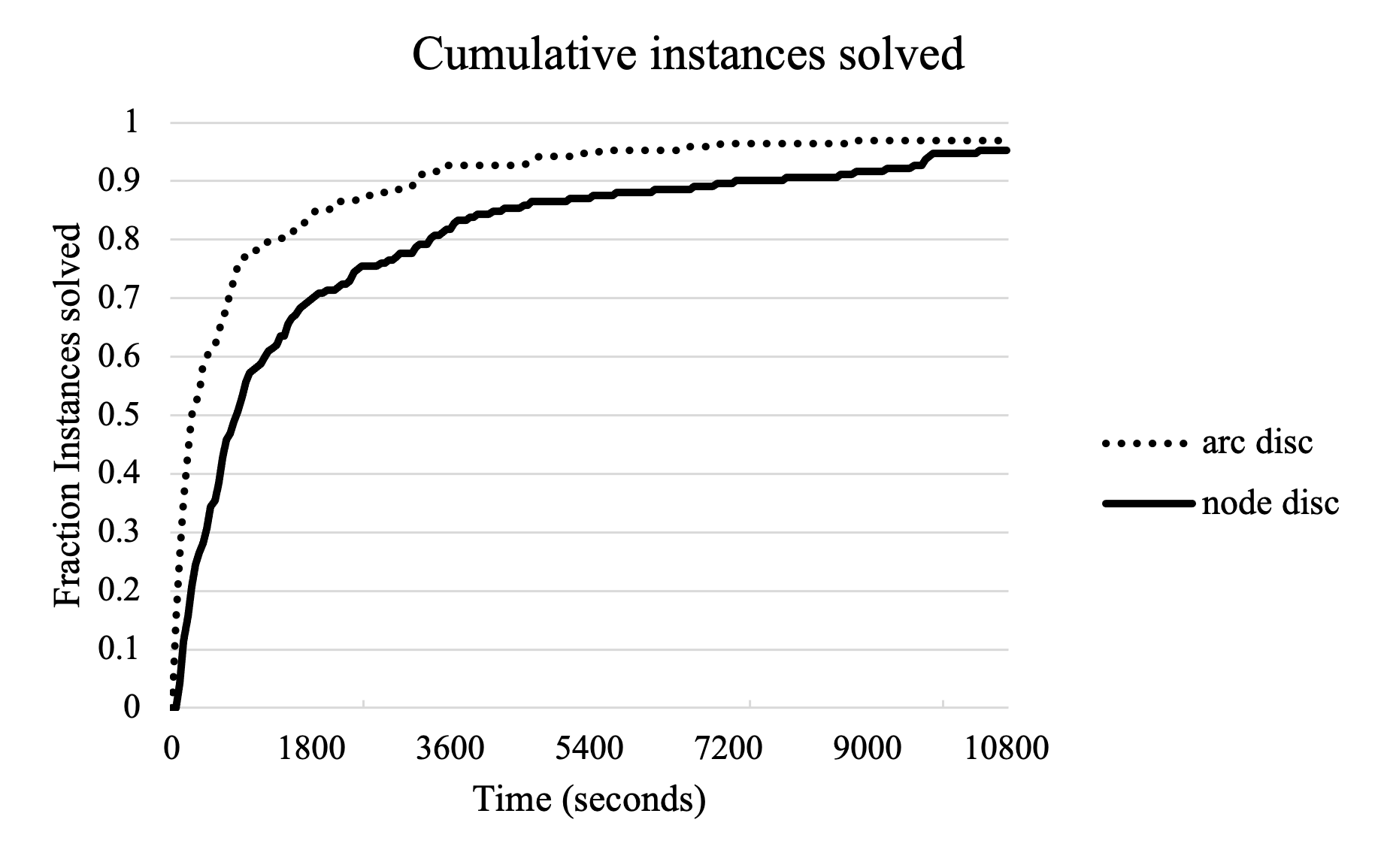}
\caption{Fraction of instances solved over time.}
\label{fig:cumulative_DP}
\end{minipage}
\begin{minipage}{0.49\textwidth}
\captionsetup{type=table}
\vspace{2cm}
\centering
\resizebox{.8\textwidth}{!}{
\begin{tabular}{l|ccc}
\toprule
 &  arc-based &  node-based  &  ratio \\ \hline
\# iterations &  24.1 &  20.4   &  1.18  \\ \hline
\# variables &  41,640 &  75,112   &  0.55 \\ \hline
\# constraints &  28,375 &  47,567  &  0.60 \\
\bottomrule
\end{tabular}}
\vspace{1.4cm}
\caption{Iteration comparison.}
\label{tab:iteration_DP}%
\end{minipage}
\end{figure}
\FloatBarrier

%%%%%%%%%%%%%%%%%%%%%%%%%%%%%%%%%%%%%%%%
%%%%%%%%%%%% Hub-and-spoke %%%%%%%%%%%%%
%%%%%%%%%%%%%%%%%%%%%%%%%%%%%%%%%%%%%%%%
\subsection{SND with hub-and-spoke networks}\label{sec:SND_regional}
In hub-and-spoke networks, regions are typically formed by selecting a hub for nearby nodes \cite{bowen}. The flat networks constructed by Crainic et al. do not allow for consistent distance-based assignment of nodes to hubs to form regions since the arc distances are selected uniformly at random. For this reason, we instead form flat networks using a geometric approach. A similar method was used to study the relaxation of node storage constraints in DDD in \cite{VanDyk}. 

In this geometric approach, $n$ nodes are randomly chosen from an $l \times l$ grid ($l = 30$). Then from this set of nodes, $h$ nodes are randomly selected and designated as hubs. Let $\H$ denote the set of hubs. Regions are formed by assigning each node to the nearest hub by L1-norm distance. We add all arcs between hub nodes and the remaining $|A| - |\H|(|\H| - 1)$ regional arcs are selected at random. We assign the transit time of the arcs in the network to be equal to the L1-norm distance between the endpoints. Note that the number of possible arcs decreases as we move to a hub-and-spoke network since all arcs are either regional or core arcs. In Section \ref{sec:SND}, the values for $|A|$ were chosen so that instances had either approximately 60$\%$ or $80\%$ of the total possible arcs. In keeping with this density, the following arc counts are chosen so that the number of arcs chosen as either 60$\%$ or $80\%$ of the expected total possible arcs in the hub-and-spoke network. The fixed costs, variable costs, and commodities are chosen as in Section \ref{sec:SND} to form flat instances. For each of the 32 flat instances with parameters in Table \ref{tab:flat_regional}, we create 6 instances of SND-RR with the parameters in Table \ref{tab:timed_regional} for a total of 192 instances. 

\begin{table}[h]
\parbox{.45\linewidth}{
\centering
\resizebox{.4\textwidth}{!}{\begin{tabular}{c|c|c|c|c|c}
\toprule
$|N|$ & $|\H|$ & $|A|$ & $|\K|$ & $F$ & $C$ \\
\hline 
20 & 3 & $70, 95$ & $100$ & 0.05, 0.1 & 1, 8\\
20 & 4 & $55, 75$ & $100$ & 0.05, 0.1 & 1, 8\\
20 & 5 & $50, 65$ & $100$ & 0.05, 0.1 & 1, 8\\
20 & 6 & $45, 60$ & $100$ & 0.05, 0.1 & 1, 8\\
\bottomrule
\end{tabular}}%
\caption{Flat instance parameters.}
\label{tab:flat_regional}
}
\hfill
\parbox{.5\linewidth}{
\centering
\vspace{.5cm}
\resizebox{.4\textwidth}{!}{\begin{tabular}{c||c|c}
\toprule
Normal distribution &  $\mu$ & $\sigma$ \\
\hline 
For generating $r_k$ & $\L$ & $\frac13\L, \frac16\L, \frac19\L $\\
For generating $p_k$ & $\frac14 \L, \frac38\L$ & $\frac16 \mu$\\
\bottomrule
\end{tabular}}%
\vspace{.5cm}
\caption{Timed instance parameters.}
\label{tab:timed_regional}
}
\end{table}

\subsubsection*{Computational results}
For hub-and-spoke instances we also observe a significant improvement in runtime when using the arc-based DDD approach. In Table \ref{tab:deciles_regional} we present the deciles for the runtime and optimality gap of each algorithm after three hours. The majority of the instances (139 out of 192) were solved to within $1\%$ of optimality within three hours by both the arc-based and node-based DDD algorithms. Among these instances, the arc-based DDD approach completed in 58.2\% of the time of the node-based approach, representing an improvement of 41.8\% in runtime. 
\begin{table}[h!]
\centering
\resizebox{.8\textwidth}{!}{
    \begin{tabular}{c|ccc|ccc}
\toprule
 & \multicolumn{3}{c}{ave. total runtime (s)} & \multicolumn{3}{c}{ave. optimality gap} \\
\midrule
Decile &  arc-based &  node-based  &  $\%$ improvement &  arc-based &  node-based  &  $\%$ improvement \\ \hline
0.1 &  29  &  62   &  54.0\% &  $<1\%$   &  $<1\%$  &  N/A   \\
0.2 &  83 &  232   &  64.2\% &  $<1\%$ &  $<1\%$ &  N/A \\
0.3 &  154 &  404   &  61.8\% &  $<1\%$ &  $<1\%$ &  N/A \\
0.4 &  307 &  832   &  63.1\% &  $<1\%$ &  $<1\%$ &  N/A \\
0.5 &  463 &  1,514   &  69.4\% &  $<1\%$ &  $<1\%$ &  N/A \\
0.6 &  1,310 &  4,096   &  57.3\% &  $<1\%$ &  $<1\%$ &  N/A \\
0.7 &  4,219 &  7,249    &  41.8\% &  $<1\%$ &  $<1\%$ &  N/A \\
0.8 &  10,727  &  10,800   & N/A  & 1.03\% &  11.0\% &  90.6\% \\
0.9 &  10,800 &  10,800   &  N/A &  19.4\% &  42.1\% &  53.9\% \\
1 &  10,800 &  10,800   &  N/A &  50.4\% &  62.6\% &  19.4\% \\
\bottomrule
\end{tabular}
}
\vspace{-.25cm}
\caption{Runtime and optimality gap comparison for hub-and-spoke instances.}
\label{tab:deciles_regional}%
\end{table}%
\FloatBarrier

The fraction of the instances solved over time is presented in Figure \ref{fig:cumulative_DP}, and we again observe a consistent improvement of the arc-based approach over the node-based approach. In Table \ref{tab:iteration_DP} we present the iteration statistics of each algorithm considering only the 139 instances which were solved within the time limit for both algorithms. The arc-based approach terminates after on average 2\% more iterations, and on average there are only 66\% the number of variables and 69\% the number of constraints in the final iteration compared to the node-based DDD approach.

\begin{figure}[h!]
\centering
\begin{minipage}{0.5\textwidth}
\centering
\includegraphics[width=\textwidth]{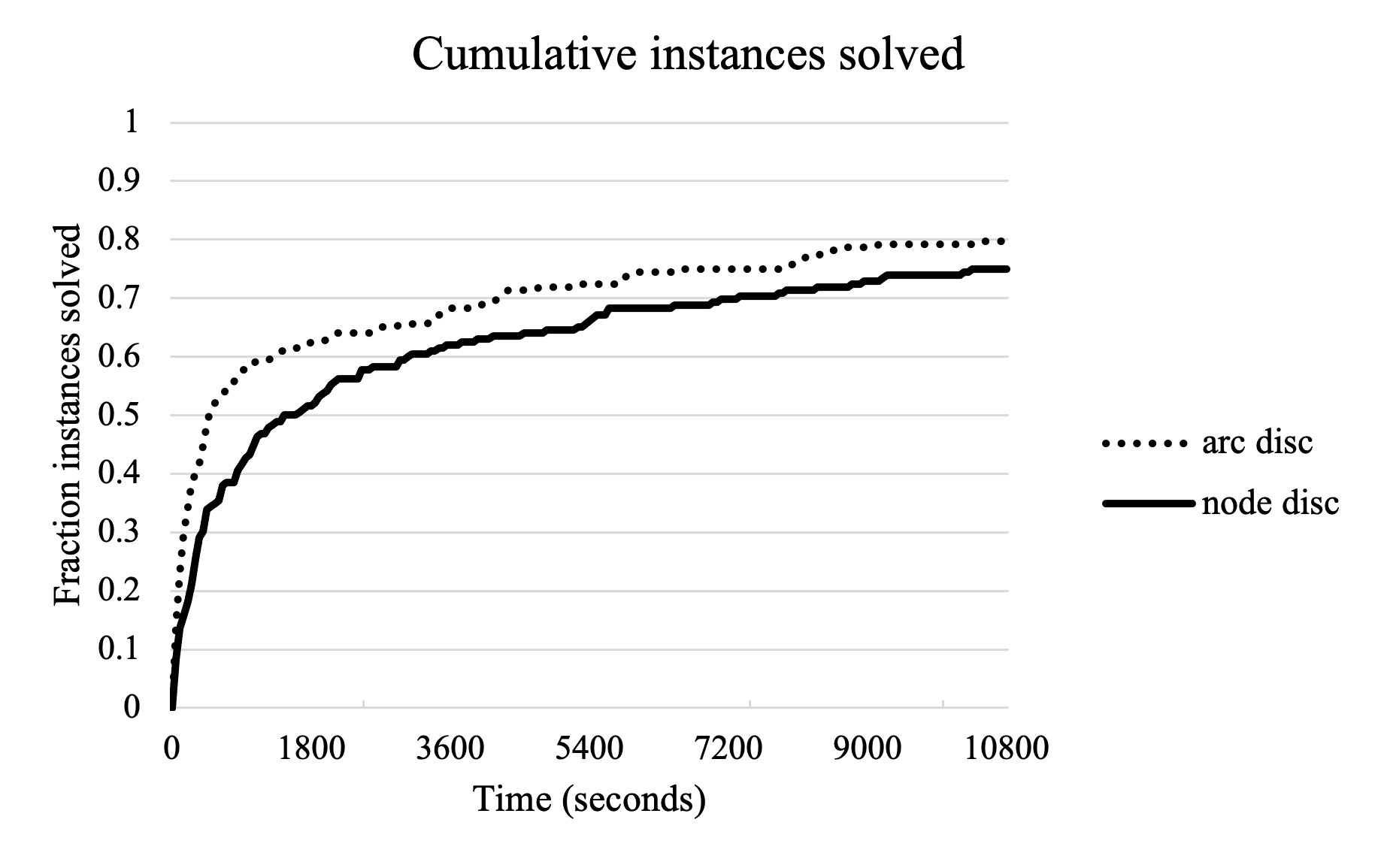}
\caption{Fraction of instances solved.}
\label{fig:figure-label1}
\end{minipage}
\begin{minipage}{0.49\textwidth}
\captionsetup{type=table} 
\vspace{2cm}
\centering
\resizebox{.8\textwidth}{!}{
\begin{tabular}{l|ccc}
\toprule
 &  arc-based &  node-based  &  ratio \\ \hline
\# iterations &  23.4 &  22.9   &  1.02  \\ \hline
\# variables &  17,102 &  26,088   &  0.66 \\ \hline
\# constraints &  8,835 &  12,715   &  0.69 \\
\bottomrule
\end{tabular}}
\vspace{1.4cm}
\caption{Iteration comparison.}
\label{tab:addlabel1}%
\end{minipage}
\end{figure}
\FloatBarrier

%%%%%%%%%%%%%%%%%%%%%%%%%%%%%%%%%%%%%%%%
%%%%%%%%%%%% SND with CPTs %%%%%%%%%%%%%
%%%%%%%%%%%%%%%%%%%%%%%%%%%%%%%%%%%%%%%%
\subsection{SND with critical times}\label{sec:SND_CPT}
To create instances with critical times, we introduce a parameter, $\alpha$, which is the number of intervals the time horizon is split into evenly. Then, for each node, a critical time is chosen uniformly at random from each of the intervals. We then modify the release times and deadlines of each commodity so that release times are rounded down to the closest critical time at that node, and deadlines are rounded up to the nearest critical time at that node. This construction allows us to capture node-dependent critical times which add additional structure to the SND instances. In total, we create 16 flat instances according to Table \ref{tab:flat_CPT} and for each flat instance we create 12 SND instances according to the parameters in Table \ref{tab:timed_CPT} and $\alpha \in \{5, 10\}$. 

\begin{table}[h!]
\parbox{.45\linewidth}{
\centering
\vspace{.25cm}
\resizebox{.4\textwidth}{!}{\begin{tabular}{c|c|c|c|c}
\toprule
$|N|$ &  $|A|$ & $|\K|$ & $F$ & $C$ \\
\hline 
20 & $230, 300$ & $150, 200$ & 0.05, 0.1 & 1, 8 \\
\bottomrule
\end{tabular}}%
\vspace{.2cm}
\caption{Flat instance parameters.}
\label{tab:flat_CPT}
}
\hfill
\parbox{.5\linewidth}{
\centering
\resizebox{.4\textwidth}{!}{\begin{tabular}{c||c|c}
\toprule
Normal distribution &  $\mu$ & $\sigma$ \\
\hline 
For generating $r_k$ & $\L$ & $\frac13\L, \frac16\L, \frac19\L $\\
For generating $p_k$ & $\frac14 \L, \frac38\L$ & $\frac16 \mu$\\
\bottomrule
\end{tabular}}%
\caption{Timed instance parameters.}
\label{tab:timed_CPT}
}
\end{table}
\FloatBarrier

\subsubsection*{Computational results}
For instances with critical times we observe a moderate improvement in runtime when using the arc-based DDD approach. In Table \ref{tab:deciles_regional} we present the deciles for the runtime and optimality gap of each algorithm after three hours. 174 out of 192 instances were solved to within $1\%$ of optimality within three hours by both the arc-based and node-based DDD algorithms. Among these instances, the arc-based DDD approach completed in 87.6\% of the time of the node-based approach, representing an improvement of 12.4\% in runtime. 
While the gap in the final row is larger in the arc-based approach, this decile only captures the maximum value. Since all the previous deciles consistently show that the arc-based approach improves over the node-based approach, we do not find this entry a convincing argument against the arc-based approach. 
\begin{table}[h!]
\centering
\resizebox{.8\textwidth}{!}{
    \begin{tabular}{c|ccc|ccc}
\toprule
 & \multicolumn{3}{c}{ave. total runtime (s)} & \multicolumn{3}{c}{ave. optimality gap} \\
\midrule
Decile &  arc-based &  node-based  &  $\%$ improvement &  arc-based &  node-based  &  $\%$ improvement \\ \hline
0.1 &  72  &  91   &  20.6\% &  $<1\%$   &  $<1\%$  &  N/A   \\
0.2 &  117 &  153   &  23.6\% &  $<1\%$ &  $<1\%$ &  N/A \\
0.3 &  192 &  247   &  20.0\% &  $<1\%$ &  $<1\%$ &  N/A \\
0.4 &  312 &  375 &  16.7\% &  $<1\%$&  $<1\%$ &  N/A \\
0.5 &  420 &  485 &  13.4\% &  $<1\%$ &  $<1\%$ &  N/A \\
0.6 &  701 &  867 &  19.2\% &  $<1\%$ &  $<1\%$ &  N/A \\
0.7 &  1,131 &  1,256 &  10.0\% &  $<1\%$ &  $<1\%$ &  N/A \\
0.8 &  2,330 &  2,775 &  16.0\% &  $<1\%$ &  $<1\%$ &  N/A \\
0.9 &  7,328 &  7,469 &  1.9\% &  $<1\%$ &  $<1\%$ &  N/A \\
1 &  10,800 &  10,800   &  N/A &  44.6\% &  36.3\% &  -22.9\% \\
\bottomrule
\end{tabular}
}
\vspace{-.25cm}
\caption{Runtime and optimality gap comparison for instances with critical times.}
\label{tab:addlabel}%
\end{table}%
\FloatBarrier 

In Table \ref{tab:iteration_CPT} we present the iteration statistics of each algorithm considering only the 174 instances which were solved within the time limit for both algorithms. The arc-based approach terminates after on average 17\% more iterations, and on average there are only 80\% the number of variables and 81\% the number of constraints in the final iteration compared to the node-based DDD approach.

\begin{figure}[h!]
\centering
\begin{minipage}{0.5\textwidth}
\centering
\includegraphics[width=\textwidth]{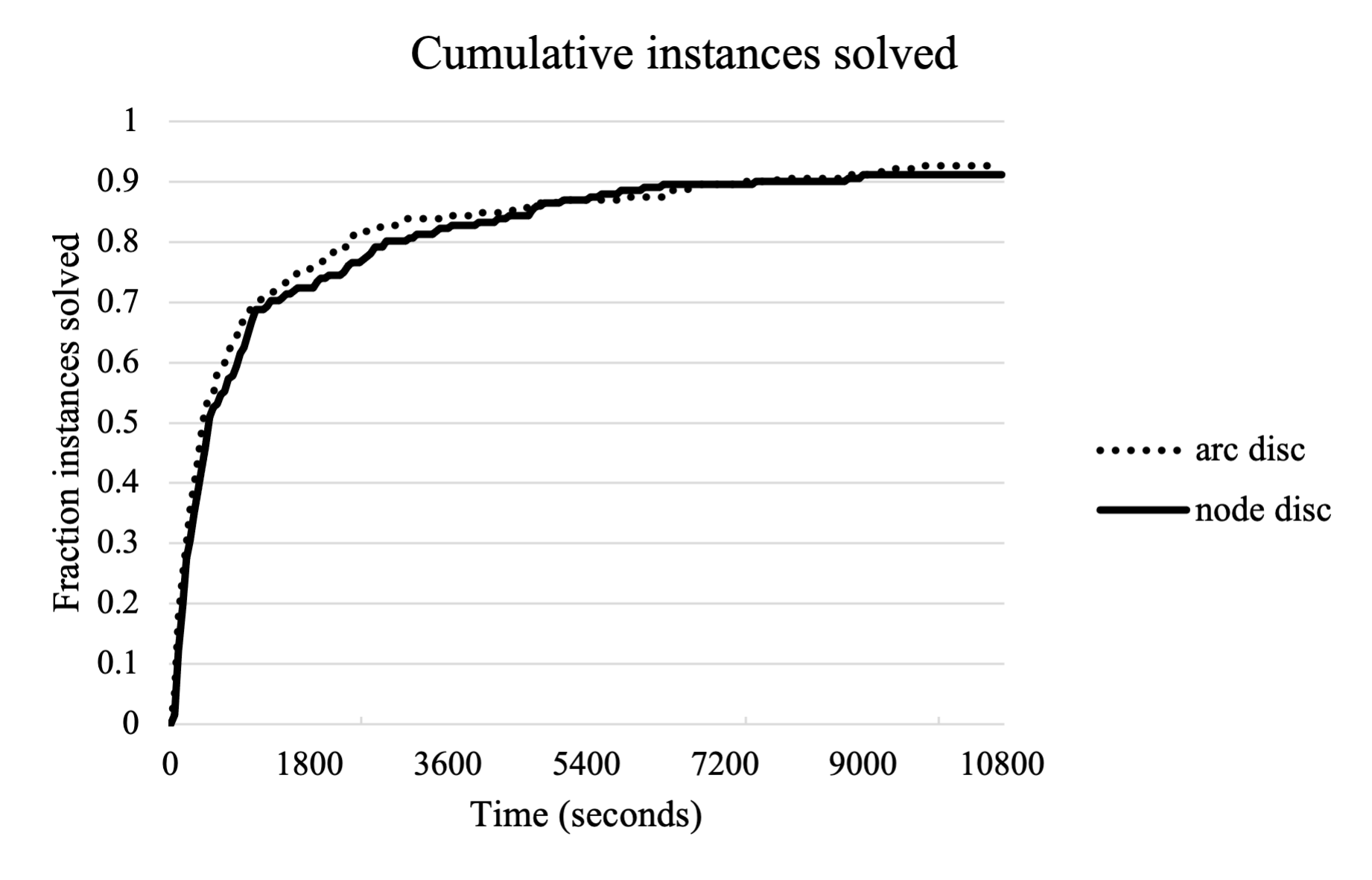}
\caption{Fraction of instances solved.}
\label{fig:figure-label}
\end{minipage}
\begin{minipage}{0.49\textwidth}
\captionsetup{type=table}
\vspace{2cm}
\centering
\resizebox{.8\textwidth}{!}{
\begin{tabular}{l|ccc}
\toprule
 &  arc-based &  node-based  &  ratio \\ \hline
\# iterations &  18.4 &  15.7  &  1.17  \\ \hline
\# variables &  34,794 &  43,697   &  0.80 \\ \hline
\# constraints &  17,696 &  21,934  &  0.81 \\
\bottomrule
\end{tabular}}
\vspace{1.4cm}
\caption{Iteration comparison.}
\label{tab:iteration_CPT}%
\end{minipage}
\end{figure}
\FloatBarrier

\section{Conclusion}\label{sec:conclusion}
Prior to our work, DDD was by design ineffective for solving problems where all near-optimal solutions require a large number of departure times at many nodes in the network. In this work, we present a novel algorithmic approach for generating smaller lower bound formulations in each iteration of DDD. This reduction in iteration size is accomplished by allowing an arc-dependent set of departure times, rather than a fixed set of departure times for each node. Moreover, the approach presented builds upon the original DDD paradigm so that the modifications required to build the arc-based DDD algorithm from a node-based DDD algorithm are minimal. The experimental results highlight the effectiveness of this DDD approach, as well as introduce important design considerations for future benchmarking instances such as regional networks and grouping release times and deadlines using critical times at nodes, which model real-world instances more closely.

The arc-based DDD approach presented in this paper offers the most improvement over the original node-based DDD approach when the arc partitioning is fine. However, for certain problems with unrestricted routes, such a fine partition of the arc set is unlikely. The discussion in the opening of Section \ref{sec:auxiliary} points to tradeoffs in allowing independent departure times for arcs in $\delta^+(v)$ in the flat network. Namely, without allowing independence of arc departure times through a valid arc partition, this splitting would result in additional timed copies of arcs in $\delta^-(v)$. When $\delta^-(v) = \emptyset$, there is no downside and the definition of a valid arc partition can be relaxed for these nodes. It is possible that when the in-degree at a node $v$ is smaller than the out-degree, allowing arc-dependent departure times is advantageous even when it results in commodities with variables for multiple copies of the same arc. Moreover, it is possible that an arc partition that changes over each iteration of DDD improves the performance of the algorithm. Characterizing when these design choices should be made is an interesting research direction. 

Finally, there are many problems for which all near-optimal solutions require a large number of departure times for the majority of the arcs. That is, solutions can only be expressed on large subgraphs of the full time-expanded network. For such problems, even the arc-based DDD approach presented here is unlikely to be find an optimal solution more quickly than solving the full time-indexed formulation. However, many temporal problems exhibit cyclic and repetitive behaviour. An interesting direction of research would be to incorporate repetitive flows into the DDD framework, without leading to an increased formulation size.

\bibliographystyle{abbrv}
\bibliography{references.bib}
\appendix

\section{Lower bound correctness for SND-RR($D_S$)}\label{app:lower_bound}
Let $\I$ be an instance of SND-RR with base graph $D = (N,A)$, commodity set $\K$, and designated subgraph $D^k = (N^k, A^k) \subseteq D$ for each $k \in \K$. Let $D_S = (N_S, A_S \cup H_S)$ be a partial network, and for each $k \in \K$ let $D_S^k$ be subgraph of $D_S$ with node set 
$N^k_S := \{(v,t) \in N_S: v \in N^k\}$
and arc set $A^k_S \cup H^k_S$ where 
$A^k_S := \{((v,t), (w,t')) \in A_S: vw \in A^k\}$ and $H^k_S := \{((v,t), (v,t'))\in H_S: v \in N^k\}$.
The following formulation is analogous to the lower bound formulation presented in Boland et al.~\cite{Boland1}.
\begin{align}\tag{SND-RR($D_S$)}\label{IP:D_S}
    \min~~ & \sum_{a \in A_S} f_a y_a + \sum_{a \in A_S} \sum_{k \in \K} c^k_a  q_k x_a^k \\
    \vspace{.5cm}
    \mbox{s.t.} ~~
	& x^k(\delta_{D_S^k}^{+}(v,t)) - x^k(\delta_{D_S^k}^{-}(v,t)) = 	\begin{cases}
        1 ~(v,t) = (o_k, r_k) \\
        -1 ~(v,t) = (d_k, l_k) \\
        0 ~\mbox{otherwise} \\
    \end{cases} \quad \forall k \in \mathcal{K}, (v,t) \in N^k_S 		\label{const1:flow}\\    
	\vspace{.5cm}
    & \sum_{k \in \mathcal{K}} q_k x_a^k \leq u_a y_a \quad \forall a \in A_S \label{const1:capacity}\\
& \sum_{a \in A_S^k} \tau_a x_a^k \leq l_k - r_k, \quad \forall k \in \K. \label{const1:feasible_D_TB}\\
    & x_a^k \in \{0,1\} \quad \forall k \in \mathcal{K}, a \in A_S^k \cup H_S^k \label{var_app1:flow} \\
	& y_a \in \mathbb{N}_{\geq 0} \quad \forall k \in \mathcal{K}, a \in A_S \label{var_app2:flow}
\end{align}
We restate properties (P1) and (P2) here for use in the following proof. 

\begin{itemize}[noitemsep]
            \item[(P1)]\textbf{Timed nodes}: For all $k \in \mathcal{K}$, $(o_k, r_k)$ and $(d_k, l_k)$ are in $N_S$;

\hspace{2.42cm} For all $v \in N$, $(v,0)$, and $(v,T)$ are in $N_S$;
            \item[(P2)]\textbf{Arc copies}: For all $a = vw \in A$, for all $(v,t) \in N_S$ with $t + \tau_{vw} \leq T$, we have
            $((v,t), (w, t')) \in A_S$ where \[t' = \max \{r: r \leq t + \tau_{vw}, (w, r) \in N_S\}.\]
        \end{itemize}

The proof of the following theorem follows from the proof of Theorem 2 in \cite{Boland1}, where we observe that the map $\mu$ defined below maps timed arcs in $D^k_T$ to timed arcs in $D^k_S$. We include the proof for completeness.
\begin{theorem}
Let $\I$ be an instance of SND-RR with base graph $D = (N,A)$, commodity set $\K$, and designated subgraph $D^k \subseteq D$ for each $k \in \K$. Let $D_S=(N_S, A_S \cup H_S)$ be a partial network that satisfies properties $(P1)$ and $(P2)$. Then an optimal solution to \ref{IP:D_S} provides a lower bound on the optimal value of \ref{IP:D_T}. 
\end{theorem}
\begin{proof}
Let $\mu: A_T \rightarrow A_S$ be the map defined so that for each timed arc $a \in A_T$, 
\begin{equation}
    a = ((v,t), (w,t')) \quad \rightarrow  \quad \mu(a) = ((v,\hat{t}), (w, \hat{t}')),
\end{equation}  where $\hat{t} = \max\{s: s \leq t, (v, s) \in N_S\}$, and $\hat{t}' = \max\{s: s \leq \hat{t} + \tau_{vw}, (v, s) \in N_S\}$. Note that $\mu$ is well-defined since $(P1)$ ensures that there is a timed copy of each node in $N$ at time 0, and $(P2)$ ensures that there is a (unique) copy of each arc in $A$ departing the selected timed node. Moreover, this timed arc underestimates the correct transit time.

Let $(\bar{x}, \bar{y})$ be a feasible solution to (\ref{IP:D_T}) with cost $C$, and let $\mathcal{Q} = \{Q_k\}_{k \in \K}$ denote the corresponding set of trajectories. Fix $k \in \K$. The set of ordered movement arcs in $Q_k$ is $\{a_1, a_2, \ldots, a_{q_k}\}$, where each $a_i = ((v_i, t_i^{out}), (v_{i+1}, t_{i+1}^{in}))$ for each $i \in [q_k - 1]$. Since $Q_k$ is feasible, the following statements are true:
\begin{enumerate}[noitemsep]
\item[(S1)] $v_1 = o_k$, and $t_1^{out} \geq r_k$
\item[(S2)] $v_{q_k} = d_k$, and $t_{q_k}^{in} \leq l_k$
\item[(S3)] $t_{i+1}^{in} = t_i^{out} + \tau_{v_i v_{i+1}}$ for all $i \in [q_k - 1]$
\item[(S4)] $t_i^{in} \leq t_i^{out}$ for all $i \in \{2, 3, \ldots, q_k - 1\}$
\end{enumerate}

Consider the ordered set of timed arcs $\{\mu(a_1), \mu(a_2), \ldots, \mu(a_{q_k})\}$ in $A_S$, where $\mu(a_i) = ((v_i, \hat{t}_i^{out}), (v_{i+1}, \hat{t}_{i+1}^{in}))$ for each $i \in [q_k - 1]$. By an analogous argument to the proof of Claim 1 in section \ref{sec:arc-based}, $\hat{t}_i^{in} \leq \hat{t}_i^{out}$ for all $i \in \{2, 3, \ldots, q_k - 1\}$. Thus, we can obtain trajectories $\hat{\mathcal{Q}} = \{\hat{Q}_k\}_{k \in \K}$ in $D_S$ by mapping each movement arc $a \in A_T$ to $\mu(a) \in A_S$, and forming trajectories by adding holdover arcs. 

Furthermore, the map $\mu$ preserves the underlying arc, and so since $Q_k$ is a trajectory in $D_T^k$, it follows that $\hat{Q}_k$ is a trajectory in $D_S^k$. By (S1) and (S2), it follows that $\hat{t}_1^{out} \geq r_k$ and $\hat{t}_{q_k} \leq l_k$. Therefore $\hat{Q}_k$ is feasible in $D_S$ for commodity $k$. 

Finally, the  underlying paths in the flat network for $\hat{\mathcal{Q}}$ and $\mathcal{Q}$ are the same for each commodity, and any pair of commodities $k, j$ traversing the same timed arc $a \in A_T$ now traverse the same timed arc $\mu(a) \in A_S$. Therefore the cost incurred when routing $\hat{\mathcal{Q}}$ in $D_S$ by for the corresponding base arc in $A$ is at most the cost incurred when routing $\mathcal{Q}$ in $D_T$, and so the cost of $\hat{\mathcal{Q}}$ is at most $C$. 
\end{proof}

\end{document}